\pdfoutput=1

\documentclass[11pt,a4paper]{article}

\usepackage[utf8]{inputenc}
\usepackage{a4wide}
\usepackage{amsmath}
\usepackage{amsfonts}
\usepackage{amssymb}
\usepackage{mathtools}
\usepackage{xcolor}
\usepackage{enumerate}
\usepackage{xspace}
\usepackage{subcaption}
\usepackage{xparse}
\usepackage{anyfontsize}
\usepackage[thmmarks, amsmath, thref, hyperref]{ntheorem}

\usepackage[pdfauthor={Sergio Waitman},%
            pdftitle={Incremental analysis of nonlinear systems with efficient methods for piecewise-affine systems},%
            pdftex,bookmarks]{hyperref}

\definecolor{dred}{rgb}{0.5,0,0}
\definecolor{red}{rgb}{0.8,0,0}
\definecolor{blue}{rgb}{0,0,0.5}
\definecolor{green}{rgb}{0,0.3,0}
\definecolor{grey}{rgb}{0.5,0.5,0.5}
		            
\usepackage{redef}
\newcommand{\paper}{paper\xspace}
\newcommand{\Spwa}{\Sigma_\textnormal{PWA}}
\newcommand{\bSig}{\overline{\Sigma}}
\newcommand{\bSpwa}{\overline{\Sigma}_\textnormal{PWA}}

\theoremstyle{break}
\newtheorem{theorem}{Theorem}[section]
\newtheorem{assumption}[theorem]{Assumption}
\newtheorem{lemma}[theorem]{Lemma}
\newtheorem{proposition}[theorem]{Proposition}
\newtheorem{definition}[theorem]{Definition}
\newtheorem{remark}[theorem]{Remark}
\newtheorem{example}{Example}

\theoremstyle{nonumberbreak}
\theoremheaderfont{\bfseries}
\theorembodyfont{\normalfont}
\theoremsymbol{\ensuremath\square}
\newtheorem{proof}{Proof}

\begin{document}

\title{Incremental analysis of nonlinear systems with efficient methods for piecewise-affine systems}

\author{S.~Waitman,
P.~Massioni,
L.~Bako,
G.~Scorletti
and V.~Fromion}

\maketitle

\begin{abstract}
	This \paper is concerned with incremental stability properties of nonlinear systems. We propose conditions to compute an upper bound on the \dLdg and to assess incremental asymptotic stability of piecewise-affine (PWA) systems. The conditions are derived from dissipativity analysis, and are based on the construction of piecewise-quadratic functions via linear matrix inequalities (LMI) that can be efficiently solved numerically. The developments are shown to be less conservative than previous results, and are illustrated with numerical examples.
	
	In the last part of this \paper, we study the connection between \dLdg stability and incremental asymptotic stability. It is shown that, with appropriate observability and reachability assumptions on the input-output operator, \dLdg implies incremental asymptotic stability. Finally, it is shown that the converse implication follows provided some regularity conditions on the state space representation are met.
\end{abstract}



\section{Introduction}

The concept of incremental stability concerns the behavior of each trajectory with respect to each other, as opposed to an equilibrium point. There exist in the literature a variety of definitions concerning incremental stability, both from the input-output and state-space points of view. Concerning the former, Zames introduced the \emph{maximum incremental amplification}~\cite{Zames1963} and used it to establish conditions for the stability of feedback loops~\cite{Zames1966a,Zames1966b}. This notion was later extended, see e.g. the \emph{generalized incremental gain}~\cite{Chitour1995} and differential stability~\cite{Georgiou1993}, and proposed as part of a framework to tackle robust performance analysis of nonlinear systems~\cite{Fromion2001a}. With respect to the latter, we may cite incremental asymptotic stability and incremental input-to-state stability~\cite{Angeli2002}, extensions of their counterparts from Lyapunov theory and Sontag's input-to-state stability, as well as convergence~\cite{Pavlov2006} and contraction~\cite{Lohmiller1998}, among some other variants. In common between these definitions is the fact that an incremental notion of stability ensures stronger properties on the behavior of the system than its non-incremental counterpart. Among these, we may cite the existence of a unique asymptotically stable constant (resp. $T$-periodic) trajectory in response to a constant (resp. $T$-periodic) input, the asymptotic independence of initial conditions and the unicity of the steady state~\cite{Angeli2002,Pavlov2006,Lohmiller1998,Fromion1997}. The aforementioned properties make incremental stability notions a suitable tool to deal with tracking and synchronization problems, as well as observer design.

In~\cite{Fromion2001a,Fromion1995}, the weighted \dLdg is proposed as a means to extend \Hoo analysis to the nonlinear context. This approach couples the quantitative characterization of performance through the addition of weighting functions with the qualitative behavior of incrementally stable systems. This enables to study robust stability and performance of nonlinear systems, and to address quantitative specifications concerning tracking/synchronization and disturbance attenuation~\cite{Fromion2003} in a way similar to that of linear time-invariant systems. The drawback of this method lies in the complexity of the conditions allowing assessment of \dLdg stability. Indeed, \cite{Romanchuk1996} proposed necessary and sufficient conditions for \dLdg stability based on the
celebrated dissipativity framework~\cite{Willems1972a}. The analysis amounts to searching for a solution of a Hamilton-Jacobi-Bellman inequality~\cite{James1993}, a problem of infinite dimension involving a partial differential inequality (PDI). Although numerical procedures to find approximate solutions to the PDI exist~\cite{James1995}, the analysis may become intractable for complex nonlinear systems. A different approach is to search for relaxed sufficient conditions to compute an upper bound on the \dLdg. In~\cite{Fromion1999}, the notion of \emph{quadratic incremental stability} is introduced, and the analysis is conducted by embedding the dynamics of the time-varying linearizations of the system in a linear parameter-varying (LPV) model with polytopic description. 
The drawback of performing analysis based on relaxed sufficient conditions comes in the form of conservativeness, and to try to cope with that we shall focus the analysis on piecewise-affine (PWA) systems.

The interest in nonlinear systems described by piecewise-affine functions is not new (see e.g.~\cite{Johansson2003} for a historical review). This may be credited to two concurring factors: 1) PWA functions allow the description of a wide range of nonlinearities appearing in applied control theory -- such as saturations, relays and dead zones -- as well as the approximation of a broad class of nonlinear functions; 2) Their description remains quite similar to that of LTI systems, so that some of the results from linear control theory can be efficiently transposed, notably with respect to the possibility of recasting the analysis as an optimization problem constrained by linear matrix inequalities (LMIs). Johansson and Rantzer pioneered the analysis of piecewise-affine systems by introducing piecewise-quadratic Lyapunov functions~\cite{Johansson1998}. In this sense, the analysis becomes local as to each region corresponds a different quadratic function. This was made possible via application of \Sproc techniques~\cite{Boyd1994}, and the approach was shown to provide less conservative results than those obtained with single quadratic Lyapunov functions. Several extensions were subsequently proposed, enabling to consider stabilization~\cite{Hassibi1998} and computation of an upper bound on the \Ldg~\cite{Rantzer2000}, among others. 

The study of incremental stability properties of PWA systems has been addressed before in the literature. In the context of convergent systems,~\cite{Pavlov2007} casts the analysis of PWA systems with continuous and discontinuous right-hand side as a search for a quadratic Lyapunov-like function. Romanchuk and Smith considered the \dLdg stability of PWA systems, and proposed conditions to construct a quadratic storage function~\cite{Romanchuk1999}. In common between both approaches is the proposal of LMI constraints and the restriction to quadratic functions. 

This \paper extends the results presented in~\cite{Waitman2016}. Based on dissipativity arguments, we propose conditions to construct storage functions and so-called incremental Lyapunov functions possessing a piecewise-quadratic structure. This is possible as the argument of these functions is not necessarily taken to be the difference between two states. Hence, the results are more general then those in~\cite{Romanchuk1999}, and thus potentially less conservative. The proposed conditions are expressed as LMI-constrained optimization problems that can be very efficiently solved by semidefinite programming solvers. 

The aim of this \paper is twofold. First, we provide sufficient conditions to assess \dLdg and asymptotic stability of piecewise-affine systems. The proposed conditions are shown to be similar, but not equivalent. The last part of this article proposes a study of the connection between both definitions, through the use of convenient observability and reachability assumptions on the state space realization of the input-output operator. The \paper is organized as follows. Section~\ref{se:IncStabNL} presents the definitions of incremental stability adopted in this \paper, along with the related functional problems allowing their assessment. In Section~\ref{se:AnalPWA} we present conditions for the incremental analysis of PWA systems, which are illustrated through numerical examples in Section~\ref{se:NumEx}. Lastly, Section~\ref{se:ObsReachCon} establishes a connection between both incremental properties at hand through some assumptions on the observability and controllability of the state representation of the nonlinear system.

\subsection*{Notation} We denote by $\norm{\cdot}$ the Euclidean norm for vectors or the corresponding induced norm for matrices. The real half line $[0,+\infty)$ is denoted by $\R_+$, and $\R_+\backslash\{0\}$ is denoted by $\R^\ast_+$. The extended real line $\R \cup \{-\infty,+\infty\}$ is denoted by $\bR$, and the half-line $\R_+\cup\{+\infty\}$ by $\bR_+$. For a vector $v = (v_1,\ldots,v_n)$, $v \succ 0$ (resp. $v \succeq 0$) is equivalent to the componentwise inequality $v_i > 0$ (resp. $v_i \geq 0$), $\forall	i \in \{1,\ldots,n\}$. For a matrix $A \in \R^{n\times n}$, $A \succ 0$ (resp. $A \succeq 0$) denotes that $A$ is positive definite (resp. semi-definite). The symbol $\bullet$ replaces the corresponding symmetric block in a symmetric matrix. The column concatenation of two matrices $A$ and $B$ of compatible dimensions, denoted by $\col$, is such that $\col(A,B) = {\text{\footnotesize$\nmatrix{c}{A\\B}$}}$. The $n \times n$ identity matrix is denoted by $I_n$, and $\bI_n \in \R^{2n\times 2n}$ and $\bJ_n \in \R^{(2n+1)\times (2n+1)}$ denote the following matrices
\begin{align}
 	\bI_n &= \nmatrix{cc}{I_n & -I_n \\ -I_n & I_n} & 
 	\bJ_n &= \nmatrix{ccc}{I_n & -I_n & 0 \\ -I_n & I_n & 0 \\ 0 & 0 & 0}
\end{align}

$\LdqRp$ is the space of square integrable $\R^q$-valued functions defined on $\R_+$, and the associated norm is defined by ${\|f\|_2=({\int\|f(t)\|^2 dt}})^{1/2}$. The causal truncation $P_T f$ is defined by $P_Tf(t)=f(t)$ for $t\leq T$ and $0$ otherwise. The {\em extended space} $\LdeqRp$ is the space of $\R^q$-valued functions defined on $\R_+$ whose causal truncations belong to $\LdqRp$ for any $T \geq 0$. The function $\phi: \R_+ \times \R_+ \times X \times \Ldep(\R_+) \rightarrow X$ is called the \emph{state transition map} and is such that $x = \phi(t,t_0,x_0,u)$ is the state $x \in X$ attained at instant $t$ when the system is driven from $x_0 \in X$ at the instant $t_0$ by the input $u$. Let $u_\tau$ denote the time shifted version of $u$, defined by $u_\tau(t) := u(t -\tau)$. A dynamical system $\Sigma$ is said to be stationary if, for every $\tau \in \R$, we have $\phi(t,t_0,\xo,u) = \phi(t+\tau,t_0+\tau,\xo,u_\tau)$.

A function $\rho: \R_+ \rightarrow \R_+$ is said to be positive definite if it is such that $\rho(0) = 0$ and $\rho(r) >0$, $\forall r \neq 0$. We denote by $\K$ the class of continuous and strictly increasing functions $\alpha: \R_+ \rightarrow \R_+$ for which $\alpha(0) = 0$. A function $\alpha$ is of class $\Koo$ if it is of class $\K$ and unbounded. A continuous function $\beta: \R_+\times\R_+ \rightarrow \R_+$ is of class $\KL$ if for any fixed $t \geq 0$, $\beta(\cdot,t) \in \K$ and, for any fixed $s$, $\beta(s,\cdot)$ is decreasing with $\lim_{t\rightarrow\infty} \beta(s,t) = 0$. The identity function is denoted by $\id$, and function composition is represented by the symbol $\circ$. For $\rho$ a function, $\rho^{k+1}$ is recursively defined by $\rho^{k+1} := \rho \circ \rho^k$, for $k \geq 1$. The floor function, denoted by $\lfloor \, \rfloor$, is such that $\lfloor r \rfloor$ is the largest integer less than or equal to $r$.

\section{Incremental stability properties}
\label{se:IncStabNL}

\subsection{Preliminaries}
\label{ss:Prel}

Let us consider an autonomous dynamical system $\Sigma_{\xo}: \LdepRp \longrightarrow \LdemRp$ with a state space representation given by
\begin{equation}
\label{eq:DynSys}
	y = \Sigma_{\xo}(u) \left\lbrace \begin{aligned} \dot{x}(t) &= f(x(t),u(t)) \\ y(t) &= h(x(t),u(t)) \\ x(0) &= x_0
	\end{aligned}\right.
\end{equation}
where $x(t)\in X \subseteq \R^n$ is the state, $u \in \LdepRp$ is the input taking values in $U = \R^p$ containing the origin, and $y \in \LdemRp$ is the output taking values in $Y = \R^m$. The functions $f:\R^n\times\R^p\rightarrow\R^n$ and $h:\R^n\times\R^p\rightarrow\R^m$ are assumed to be Lipschitz continuous, and are such that $f(0,0)=0$ and $h(0,0)=0$, so that the origin is an equilibrium point associated to the null input with zero output. For the sake of notation, we shall drop the subscript from $\Sigma_{\xo}$ and note simply $\Sigma$. It should be noted, however, that the initial condition has an impact on the input-output behavior of $\Sigma$, which will be important in what follows.

In this \paper, we shall be concerned with establishing efficient methods to assess incremental properties of system~\eqref{eq:DynSys}. We shall consider two incremental stability notions: \dLdg stability and incremental asymptotic stability. The former is an input-output property, based on the energy relation of any two different outputs with respect to that of the respective inputs, whereas the latter concerns the behavior of state trajectories with respect to initial conditions when the system is driven by a certain input.

As it will be elaborated in the sequel, the assessment of both of these incremental properties might prove too difficult when dealing with general nonlinear systems. Indeed, both may be obtained by solving related functional problems, a task which is far from trivial in most cases. For this reason, the primary aim of this \paper it to provide sufficient conditions to assess incremental properties for a class of nonlinear systems, namely piecewise-affine (PWA) systems. Let us consider a dynamical system $\Spwa: \LdepRp \rightarrow \LdemRp$ with a PWA representation given by
\begin{equation}
\label{eq:PWAsysUY}
	y \!=\! \Spwa(u)\!\left\{
	\begin{aligned} &
		\begin{aligned}
			\dot{x}(t) &= A_ix(t) + a_i + B_iu(t) \\
				  y(t) &= C_ix(t) + c_i + Du(t) \\			  
		\end{aligned} & \text{for } \xt\in X_i \\
		& x(0) = x_0
	\end{aligned}
	\right.		
\end{equation}
where the regions $X_i$, for $i \in \I := \{1, \ldots, N\}$, are closed convex polyhedral sets defined by $X_i = \{ x \in X \mid G_ix + g_i \succeq 0 \}$ with non-empty and pairwise disjoint interiors such that $\bigcup_{i\in \I} X_i= X$. Then, $\{X_i\}_{i \in \I}$ constitutes a finite partition of $X$. From the geometry of $X_i$, the intersection $X_i \cap X_j$ between two different regions is always contained in a hyperplane, i.e. $X_i \cap X_j \subseteq \left\{ x \in X \mid E_\ij x + e_\ij = 0\right\}$.

We require that~\eqref{eq:PWAsysUY} respects the following assumptions.

\begin{assumption}
\label{ass:ZeroEqui}
	For any $i \in \I$, $0 \in X_i$ implies $a_i = 0$ and $c_i = 0$, so that $x = 0$ is an equilibrium point of~\eqref{eq:PWAsysUY} with zero input and zero output.
\end{assumption}

\begin{assumption}
\label{ass:NoSlidMod}
	The PWA system~\eqref{eq:PWAsysUY} does not present Zeno behavior.
\end{assumption}

We say that the system presents Zeno behavior when there is an infinite number of switches in a bounded interval of time. There exist two types of so-called Zeno behavior, namely \emph{chattering Zeno} and \emph{genuine Zeno}~\cite{Ames2005}. The former is characterized by a null dwell-time, when the trajectory evolves on a switching surface, which is generally called a \emph{sliding mode}.  The latter, on the contrary, is characterized by strictly positive dwell-times, which tend to zero as the time approaches the so-called breaking Zeno point. A sufficient condition to ensure the non-existence of sliding modes is the Lipschitz continuity of the right-hand side of the differential equation in~\eqref{eq:PWAsysUY}. Conditions to ensure Lipschitz continuity are given in Appendix~\ref{ss:LipCont}. In the case of discontinuous right-hand side, it might be difficult to ensure non-existence of sliding modes in the general case (see e.g.~\cite{Johansson2003}). Genuine Zeno behavior is somewhat more complicated to exclude. \cite{Thuan2014} states that no Zeno behavior exist in the case where the right-hand side of the ODE is continuous and the input is piecewise real-analytic. \cite{Khan2015} extends the results to ensure non-existence of Zeno behavior when the right-hand side of the ODE is continuous and the input is left/right-analytic.

In the remaining of this section we shall formally define the incremental properties we aim to study, and present the general approach to establish them. We then propose the main results in Section~\ref{se:AnalPWA}.

\subsection{Incremental \Ldg and dissipativity}
\label{ss:IncL2Diss}

In this section we present the adopted definition of \dLdg together with the associated machinery from the framework of dissipative systems~\cite{Willems1972a}.  We begin by recalling the definition of \Ldg stability of nonlinear systems.

\begin{definition}[\Ldg stability]
\label{def:L2}
	The system~\eqref{eq:DynSys} is said to be \Ldg stable if there exists $\gamma \geq 0$ such that for all $u \in \LdpRp$ we have
	\begin{equation}
	\label{eq:L2cond}
	 	\int_{0}^\infty \! \norm{\yt}^2\,dt \leq \gamma^2 \! \int_{0}^\infty \! \norm{\ut}^2 \,dt
	\end{equation}
	for $y = \Sigma(u)$ with $x_0 = 0$. We define the \Ldg of $\Sigma$ as the smallest $\gamma$ for which~\eqref{eq:L2cond} holds.
\end{definition}

We now proceed to define the \dLdg of a dynamical system.

\begin{definition}[Incremental \Ldg stability]
\label{def:dL2}
	The system~\eqref{eq:DynSys} is said to be incrementally \Ldg stable if it is \Ldg stable and there exists $\eta \geq 0$ such that for all $u$, $\tu\in\LdpRp$ we have
	\begin{equation}
	\label{eq:IncL2cond}
		\int_{0}^\infty \! \norm{\yt - \tyt}^2\,dt \leq \eta^2 \! \int_{0}^\infty \! \norm{\ut - \tut}^2 \,dt
	\end{equation}
	for $y = \Sigma(u)$ and $\ty = \Sigma(\tu)$ with the same initial condition $x_0$. We define the \dLdg of $\Sigma$ as the smallest $\eta$ for which~\eqref{eq:IncL2cond} holds.
\end{definition}

We note that, given the assumptions on functions $f$ and $h$, boundedness  of the \Ldg is implied by \dLdg stability. We shall make use of the framework of dissipative systems~\cite{Willems1972a}, a standard procedure when studying input-output properties such as boundedness and passivity. The following recalls the main concepts needed.

\begin{definition}[Dissipative system]
	A dynamical system $\Sigma$ is said to be dissipative with respect to the \emph{supply rate} $w: U \times Y \rightarrow \R$ if there exists a nonnegative function $S: X \longrightarrow \R_+$, called the \emph{storage function}, such that for all $t_1, t_0 \in \R_+, \, t_1 \geq t_0$, and $u \in \LdepRp$,
	\begin{equation}
	\label{eq:DI}
	 	S(x(t_1)) - S(x(t_0)) \leq  \int_{t_0}^{t_1} \! w(u(t),y(t)) \; dt 
	\end{equation}
	where $x(t_1) = \phi(t_1,t_0,x(t_0),u)$ and $y = \Sigma(u)$. 
\end{definition}	

The incremental \Ldg stability of system~\eqref{eq:DynSys} can be assessed via dissipativity analysis of a fictitious augmented system $\bSig: \LdepRp \times \LdepRp \rightarrow \LdemRp$ given by

\begin{equation}
\label{eq:DynSysAug}
	\by = \bSig(u,\tu) \left\lbrace \begin{aligned} \dxt &= f(\xt,\ut) \\ \dtxt &= f(\txt,\tut) \\ \byt &= h(\xt,\ut) - h(\txt,\tut) \\
		 x(0) &= x_0 \\
		 \tx(0) &= \tx_0 \\
	\end{aligned}\right.
\end{equation}

We note that $\bSig(u,\tu) := \Sigma(u) - \Sigma(\tu)$. Before stating the connection between \dLdg stability and dissipativity of the augmented system, we need to recall the definition of reachability, as follows.

\begin{definition}[Reachability]
	The state space of $\Sigma$ is said to be reachable from $x_0$ if given any $x \in X$ and $t \geq 0$, there exist $u \in \LdepRp$ and $T_r \geq 0$ such that $x = \phi(t,t-T_r,x_0,u)$.
\end{definition}

The next theorem is an important result concerning the equivalence between \dLdg stability and dissipativity of the augmented system.

\begin{theorem}[Incremental dissipativity]
\label{th:IncDiss}
	 Let $\Sigma$ be a stationary dynamical system with a reachable state space $X$ from $x_0$. The following statements are equivalent:
	 \begin{enumerate}[(i)]
	 	\item $\Sigma$ is incrementally \Ldg stable;
	 	\item $\bSig$ is dissipative with respect to the supply rate function $\bw: U \times U \times Y \rightarrow \R$ given by
	 		\begin{equation}
			\label{eq:IncSupRate}
				\bw(u,\tu,\by) = \eta^2\norm{u - \tu}^2 - \norm{\by}^2
			\end{equation}
			and there exists a storage function $S$, defined from $X \times X$ into $\R_+$, i.e., independent of time, such that $S(x,x) = 0$, $\forall x \in X$.
	 \end{enumerate}
\end{theorem}

\begin{proof}
	Taken from~\cite{FromionBook}. It is a consequence of Lemmas 3.1 and 3.2 in~\cite{Fromion1997}.
\end{proof}

\begin{remark}
	A similar treatment is presented in~\cite{Romanchuk1996}, albeit with some minor differences. If reachability is not assumed, condition (ii) becomes only sufficient to obtain (i).
\end{remark}

For the sake of notation, we shall say that $\Sigma$ is \emph{incrementally dissipative} when the equivalent statements in Theorem~\ref{th:IncDiss} are satisfied.

We aim to find sufficient conditions to compute an upper bound on the \dLdg, and we focus on the class of PWA systems. This problem was tackled before in~\cite{Romanchuk1999}, where the search for a quadratic storage function of the form $S(x,\tx) = (x - \tx)^TP(x - \tx)$ is expressed as an LMI-constrained optimization problem. We shall propose less conservative results based on the search for piecewise-quadratic storage functions of $x$ and $\tx$, as will be presented in Section~\ref{se:AnalPWA}.

\subsection{Incremental asymptotic stability}
\label{ss:IncAsStab}

We now consider the notion of incremental asymptotic stability. As is the case with its non-incremental counterpart, we study the influence of the initial condition on the trajectories of the system, but now subjected to a certain input belonging to a specific class. The following definition is adapted from~\cite{Angeli2002}.

\begin{definition}[Incremental asymptotic (exponential) stability]
\label{def:dAS}
	Let $\X \subseteq X$ be a non-empty connected region of the state space, and let $\U \subseteq \LdepRp$ be a class of input functions. We say that system~\eqref{eq:DynSys} is incrementally asymptotically stable on $\X$ with respect to $\U$ if there exists a function $\beta$ of class $\KL$ so that for all $x_0, \tx_0 \in \X$ and all $t \geq 0$ the following holds
	\begin{equation}
	\label{eq:dGAS}
		\norm{x(t) - \tx(t)} \leq \beta(\norm{x_0 - \tx_0}, t)
	\end{equation}
	with $\xt = \phi(t,0,x_0,u)$ and $\txt = \phi(t,0,\tx_0,u)$ for any $u \in \U$. If there exist $d,\lambda > 0$ such that $\beta(r,t) \leq de^{-\lambda t}r$, the system is said to be incrementally exponentially stable on $\X$ with respect to $\U$. If $\X = X = \R^n$, the system is said to be incrementally globally asymptotically (exponentially) stable with respect to $\U$.
\end{definition}

\begin{remark}
	The notion of incremental asymptotic stability, as stated here, can be related to the so called \emph{asymptotic stability of the unperturbed motion}, see e.g.~\cite{Fromion1997}.
\end{remark}

It is interesting to note the relation between incremental asymptotic stability and its non-incremental counterpart. The usual definitions of asymptotic stability of~\eqref{eq:DynSys} concern the stability of the origin subject to the state equation $\dxt = f(\xt,0)$, i.e. associated to the input $u \equiv 0$. Incremental asymptotic stability, on the other hand, is concerned with the stability of any equilibrium point or trajectory associated to a certain input $u$.

In~\cite[Theorem 1]{Angeli2002}, a slightly different definition of incremental asymptotic stability is shown to be equivalent to the existence of what we may call an incremental Lyapunov function. For completeness, we restate this theorem here as a sufficient condition in view of the definition adopted in the present \paper.

\begin{theorem}
\label{th:dAS}
	System~\eqref{eq:PWAsysUY} is incrementally asymptotically stable with respect to $\U$ as in Definition~\ref{def:dAS} if there exist a continuous function $V: X \times X \rightarrow \R_+$, called an incremental Lyapunov function, and $\Koo$ functions $\alpha_1$ and $\alpha_2$ such that
	\begin{equation}
	\label{eq:dASnorm}
		\alpha_1\big(\!\norm{x - \tilde{x}}\!\big) \leq V(x,\tilde{x}) \leq \alpha_2\big(\!\norm{x - \tilde{x}}\!\big)
	\end{equation}
	and along any two trajectories $x,\tx$, starting respectively from $x_0$ and $\tx_0$ under input $u \in \U$, $V$ satisfies for any $t \geq 0$
	\begin{equation}
	\label{eq:dASnegdef}
	 	V(\xt,\txt) - V(x_0,\tx_0) \leq -\int_0^t\!\rho\big(\!\norm{x(\tau) - \tx(\tau)}\!\big) \, d\tau
	\end{equation}
	with $\xt = \phi(t,0,x_0,u)$, $\txt = \phi(t,0,\tx_0,u)$ and $\rho$ a positive definite function.
\end{theorem}

Hence, to assess incremental asymptotic stability, it suffices to find a function $V$ respecting the conditions in Theorem~\ref{th:dAS}. However, as in the previous section, the difficulty remains how to efficiently construct this function in the general case of nonlinear systems. The problem of finding efficient methods to assess incremental asymptotic stability is largely open. In the next section we shall propose a method to efficiently construct piecewise-quadratic incremental Lyapunov functions.

\section{Incremental analysis of piecewise-affine systems}
\label{se:AnalPWA}

In this section we shall focus on the incremental analysis of PWA systems. The augmented system~\eqref{eq:DynSysAug} may be represented as follows when the system is piecewise-affine.

\begin{equation}
\label{eq:PWAsysUYaug}
	\by = \bSpwa(\bu) \left\lbrace \begin{aligned} & \begin{aligned} \dbxt &= \bA_\ij\bxt + \bB_\ij\but \\ \byt &= \bC_\ij\bxt + \bD\but \end{aligned} & \text{for } \bxt\in X_\ij \\
		& \bx(0) = \bx_0
	\end{aligned}\right.
\end{equation}
where $\bx = \col(x,\tx,1)$, $\bu = \col(u,\tu)$, and
\begin{equation}
\label{eq:AugMat}
\begin{aligned}
	\bA_\ij &= \nmatrix{ccc}{A_{i} & 0 & a_{i} \\ 0 & A_{j} & a_{j} \\ 0 & 0 & 0} &
	\bB_\ij &= \nmatrix{cc}{B_{i} & 0 \\ 0 & B_{j} \\ 0 & 0}
	\\
	\bC_\ij &= \nmatrix{ccc}{C_{i} & -C_{j} & c_{i} - c_{j}} &
	\bD &= \nmatrix{cc}{D & -D}
\end{aligned}
\end{equation}

The space $\bX$ is defined as $\bX = X \times X \times \{1\}$, and regions $X_\ij$ are defined as $X_\ij = \{\bx \in \bX \mid x \in X_i \text{ and } \tx \in X_j\}$. Each region $X_\ij$ is described by $X_\ij = \{ \bx \in \bX \mid \bG_\ij \bx \succeq 0 \}$ where
\begin{equation}
 	\bG_\ij =  \nmatrix{ccc}{G_{i} & 0 & g_{i} \\ 0 & G_{j} & g_{j}}
\end{equation}

Analogously to the state partition $\{X_i\}_{i \in \I}$ of system $\Spwa$, the intersection between any two regions $X_\ij$ and $X_{kl}$ of $\bSpwa$ is either empty or contained in the hyperplane given by
\begin{equation}
\label{eq:AugSysHypP}
	X_\ij \cap X_{kl} \subseteq \left\{ \bx \in \bX \mid \bE_{ijkl} \bx = 0\right\}
\end{equation}

In the next section we establish sufficient conditions for \dLdg stability and incremental asymptotic stability. These problems will be cast as the search for continuous piecewise-quadratic functions of the form
\begin{equation}
\label{eq:PWQStorFunc}
	S(x,\tx) = \begin{cases}
 			(x - \tx)^TP_i(x - \tx) & \!\text{for } \bx \in X_{ii} \\
 			\bx^T\bP_\ij\bx & \!\text{for } \bx \in X_\ij, \, i \neq j
 		\end{cases}
\end{equation}

We remark that the choice of a quadratic function of $(x - \tx)$ on cells $X_{ii}$ does not lead to any loss of generality with respect to general piecewise-quadratic functions of $\bx$. Indeed, from Theorem~\ref{th:IncDiss}, if $x \in X$ is reachable from $\xo$, we have that $S(x,x) = 0$, which implies the aforementioned structure. This is formally stated in the following proposition.

\begin{proposition}
\label{prop:SXii}
	If $S: X \times X \rightarrow \R_+$ is a nonnegative piecewise-quadratic function given by $S(x,\tx) = \bx^T\bP_\ij\bx$ for $\bx \in X_\ij$ and such that $S(x,x) = 0$ for all $x \in X$, then on regions $X_{ii}$, $S$ must be of the form $S(x,\tx) = (x - \tx)^TP_i(x - \tx)$.
\end{proposition}

\begin{proof}
	Let us denote
	\begin{equation}
		\bP_\ij = \nmatrix{ccc}{P_\ij^{11} & P_\ij^{12} & q_\ij^{1}\Tstrut\Bstrut \\ \bullet & P_\ij^{22} & q_\ij^{2}\Tstrut\Bstrut \\ \bullet & \bullet & r_\ij\Tstrut\Bstrut}
	\end{equation}
	We have that $S(x,x) = 0$, for all $x \in X$. Then, in regions $X_{ii}$ we have
	\begin{equation}
	 	S(x,x) \!=\! x^T\!\left(P_{ii}^{11} \!+\! 2P_{ii}^{12} \!+\! P_{ii}^{22}\right)\!x + 2\left(q_{ii}^{1} \!+\! q_{ii}^{2}\right)\!x + r_{ii} \!=\! 0
	\end{equation}
	for all $x \in X_i$, which implies that
	\begin{align}
		P_{ii}^{12} &= -\frac{1}{2}\left(P_{ii}^{11} + P_{ii}^{22}\right) \\
		q_{ii}^{1} &= -q_{ii}^{2} =: q_{ii} \\
		r_{ii} &= 0
	\end{align}
	The function $S$ can then be rewritten as
	\begin{align}
		S(x,\tx) \!=\! (x \!-\! \tx)^TP_{ii}^{11}(x \!-\! \tx) - x^TQ(x \!-\! \tx) + 2q_{ii}^T(x \!-\! \tx)
	\end{align}	 
	for $\bx \in X_{ii}$, with $Q = P_{ii}^{22} - P_{ii}^{11}$. 
	
	Since regions $X_i$ have non-empty interiors, the vector given by $x - \tx$ for $x,\tx \in X_i$ can take any direction on $\R^n$. Based on this fact, we can take $x, \tx_1, \tx_2 \in X_i$ such that $x - \tx_1 = \alpha\xi$ and $x - \tx_2 = -\alpha\xi$, for $\alpha > 0$ and some $\xi \in \R^n$. We can then write
	\begin{equation}
	\begin{aligned}
	S(x,\tx_1) &= \alpha^2\xi^TP_{ii}^{11}\xi - \alpha(x^TQ\xi - 2q_{ii}^T\xi) \geq 0 \\
	S(x,\tx_2) &= \alpha^2\xi^TP_{ii}^{11}\xi + \alpha(x^TQ\xi - 2q_{ii}^T\xi) \geq 0
	\end{aligned}
	\end{equation}
	For $\norm{\xi} \rightarrow 0$, the quadratic term becomes negligible and the sign is dominated by the remaining term. Since $S$ is nonnegative and $\alpha > 0$, we must have $x^TQ\xi = 2q_{ii}^T\xi$, $\forall \xi \in \R^n$ and $\forall x \in X_i$. Since $\xi$ is an arbitrary vector in $\R^n$, there exists an $n \times n$ matrix $\Xi$ of full rank so that $(Qx - 2q_{ii})^T\xi = 0 $ implies $ (Qx - 2q_{ii})^T\Xi = 0$ and then 
	\begin{equation}
	\label{eq:PosAuxEq}
	 	Qx = 2q_{ii}, \quad \forall x \in \inte{X_i}
	\end{equation}
	Every vector $x \in \inte{X_i}$ can be written as $x = x_0 + (x - x_0) = x_0 + \alpha{}\xi$, for some $x_0 \in X_i$, $\alpha >0$ and $\xi \in \R^n$. Substituting in~\eqref{eq:PosAuxEq}, we get $Qx_0 + \alpha Q\xi = 2q_{ii}$. Using~\eqref{eq:PosAuxEq} yields $\alpha Q\xi = 0$. Since $\xi$ can take any direction, this requires that the null space of $Q$ be of dimension $n$, which implies $Q = 0$. Therefore, $P_i := P_{ii}^{11} = P_{ii}^{22}$ and $q_{ii} = 0$, and function $S$ becomes
	\begin{equation}
	\label{eq:PWQPosDefXii}
	 	S(x,\tx) = (x - \tx)^TP_i (x - \tx) \quad \text{for } x,\tx \in X_i
	\end{equation}
	which concludes the proof.
\end{proof}

\begin{remark}
\label{rem:StatPart}
	The state space partition for the construction of the storage function $S$ in~\eqref{eq:PWQStorFunc} was chosen to be the same as $\{X_\ij\}_{(i,j) \in \I \times \I}$. However, in general, both are independent and the state space partition for the storage function may be refined to allow for more flexibility (see e.g.~\cite{Johansson1998}).
\end{remark}

\subsection{Piecewise quadratic incremental \Ldg stability}
\label{ss:PWQdL2S}

In this section we formulate conditions allowing to compute an upper bound on the \dLdg based on the search for a continuous piecewise-quadratic storage function~\eqref{eq:PWQStorFunc}.

\begin{theorem}
\label{th:PWAdL2S}
	If there exist symmetric matrices $P_{i} \in \R^{n \times n}$ and $\bP_{ij} \in \R^{(2n+1) \times (2n+1)}$; $U_{ij}, W_{ij} \in \R^{p_{ij} \times p_{ij}}$ with nonnegative coefficients and zero diagonal; and $L_{ijkl} \in \R^{(2n+1)\times 1}$ such that
	\begin{align}
		& \begin{cases}
			P_{i} \succeq 0 \\
			\nmatrix{cc}{A_i^TP_{i} + P_{i}A_i + C_i^TC_i & P_iB_i + C_i^TD \\ \bullet & D^TD - \eta^2I_p} \preceq 0
		\end{cases} & &\text{for } i \in \I \label{eq:ContThm1_dL2S}\\
		& \begin{cases}
			\bP_{ij} - \bG_{ij}^TU_{ij}\bG_{ij} \succeq 0 \\
			\nmatrix{cc}{
				\bA_\ij^T\overline{P}_\ij + \overline{P}_\ij \bA_\ij + \bC_\ij^T\bC_\ij + \bG_{ij}^TW_{ij}\bG_{ij} & 
				\bP_\ij\bB_\ij + \bC_\ij^T\bD \\
				\bullet &
				\bD^T\bD - \eta^2\overline{I}_p} \preceq 0
		\end{cases}  & &\begin{aligned}\text{for } (i,j) \in \I \times \I, \\ i \neq j\end{aligned} \label{eq:ContThm2_dL2S} \\
		& \bP_{ij} = \bP_{kl} + L_{ijkl}\bE_{ijkl} + \bE_{ijkl}^TL_{ijkl}^T & & \begin{aligned}\text{for } (i,j),(k,l), \\ X_\ij \cap X_{kl} \neq \varnothing\end{aligned} \label{eq:ContThm3}
	\end{align}
	 are satisfied, then the piecewise-affine system~\eqref{eq:PWAsysUY} is incrementally \Ldg stable and has an incremental \Ldg less than or equal to $\eta$. Furthermore, it is incrementally dissipative with storage function $S$ given by~\eqref{eq:PWQStorFunc} and supply rate given by~\eqref{eq:IncSupRate}.
\end{theorem}

\begin{proof}
	According to Theorem~\ref{th:IncDiss}, the \dLdg of~\eqref{eq:PWAsysUY} is less than or equal to $\eta$ if the augmented system~\eqref{eq:PWAsysUYaug} is dissipative with respect to the supply rate~\eqref{eq:IncSupRate}. We will show that the LMIs~\eqref{eq:ContThm1_dL2S},~\eqref{eq:ContThm2_dL2S} and the matrix equality~\eqref{eq:ContThm3} allow the construction of a continuous nonnegative piecewise-quadratic storage function $S$ of structure given by~\eqref{eq:PWQStorFunc} such that the above condition is met.
	
	\noindent\emph{Continuity} -
	We first show that $S$ is a continuous function of $\bx$. This is clearly the case inside every cell, so we just need to show continuity on the boundaries. From~\eqref{eq:AugSysHypP}, $\bE_{ijkl}\bx = 0$ for all $\bx \in X_\ij \cap X_{kl}$, then~\eqref{eq:ContThm3} implies that $ \bx^T\bP_\ij\bx = \bx^T\bP_{kl}\bx$ for $\bx \in X_\ij \cap X_{kl}$ and hence that $S$ is continuous.
	
	\noindent\emph{Nonnegativity} -
	We now show that $S$ is a nonnegative function. The first inequality in~\eqref{eq:ContThm2_dL2S}, post and pre multiplied respectively by $\bx$ and $\bx^T$, implies that $\bx^T\bP_\ij\bx \geq \bx^T\bG_{ij}^TU_{ij}\bG_{ij}\bx$. Since $U_{ij}$ is composed of nonnegative coefficients, the right-hand side of the previous inequality is nonnegative whenever $\bx \in X_\ij$. This implies that
	\begin{equation}
	\label{eq:PWQPosDefXij_dL2S}
	 	\bx^T\bP_\ij\bx  \geq 0 \qquad \text{for } \bx \in X_{ij}
	\end{equation}
	
	The first inequality in~\eqref{eq:ContThm1_dL2S} implies that $S(x,\tx) \geq 0$ for all $\bx \in X_{ii}$. With~\eqref{eq:PWQPosDefXij_dL2S}, this guarantees that	
	\begin{equation}
	\label{eq:PWQPosDef_dL2S}
		S(x,\tx) \geq 0, \quad\forall x,\tx \in X
	\end{equation}
	
	\noindent\emph{Dissipation inequality} -
	We now show that the storage function respects the dissipation constraint~\eqref{eq:DI}. Using the same arguments as before, the last inequality in~\eqref{eq:ContThm2_dL2S}, post and pre multiplied by $\col(\bx, \bu)^T$ and $\col(\bx, \bu)$, implies that
	\begin{equation}
		\bx^T\bP_\ij(\bA_\ij\bx + \bB_\ij\bu)  + (\bA_\ij\bx + \bB_\ij\bu)^T\bP_\ij\bx \;+(\bC_\ij\bx + \bD\bu)^T(\bC_\ij\bx + \bD\bu) - \eta^2\bu^T\overline{I}_p\bu \leq 0
	\end{equation}
	for all $\bu \in U \times U$ and all $\bx \in X_\ij$. Let $t_a$ and $t_b$ be two time instants such that the state trajectory of system~\eqref{eq:PWAsysUYaug} remains in $\Xij$ on the interval $[t_a,t_b]$. By noticing that $\dbx = \bA_\ij\bx + \bB_\ij\bu$, and integrating from $t_a$ to $t_b$ along a trajectory of~\eqref{eq:PWAsysUYaug}, we have
	\begin{equation}
	\label{eq:intSij_dL2S}
		\!\!\bx(t_b)^T\bP_\ij\bx(t_b) - \bx(t_a)^T\bP_\ij\bx(t_a) +  \!\!\int_{t_a}^{t_b}\! \norm{\ytau \!-\! \tytau}^2d\tau - \eta^2\!\!\int_{t_a}^{t_b}\! \norm{\utau \!-\! \tutau}^2\,d\tau \leq 0
	\end{equation}
	
	The same reasoning can be applied to the last inequality in~\eqref{eq:ContThm1_dL2S}, post and pre multiplying by $\col(x - \tx,u - \tu)^T$ and $\col(x - \tx,u - \tu)$, which yields 
	\begin{equation}
	\label{eq:intSii_dL2S}
		\Delta x(t_b)^TP_i\Delta x(t_b) - \Delta x(t_a)^TP_i\Delta x(t_a) + \int_{t_a}^{t_b}\! \norm{\ytau \!-\! \tytau}^2d\tau - \eta^2\!\!\int_{t_a}^{t_b}\! \norm{\utau \!-\! \tutau}^2\,d\tau \leq 0
	\end{equation}
with $\Delta x(t) := \xt - \txt$. We note that the first terms in~\eqref{eq:intSij_dL2S} and~\eqref{eq:intSii_dL2S} represent the storage function~\eqref{eq:PWQStorFunc}. Let us consider a trajectory $\bx(t)$, $\forall t \in [t_0,t_1]$, with $t_0 \geq 0$. The time $t_1$ can be decomposed as $t_1 = t_1 - t_{in,q} + \sum_{k=0}^{q-1} (t_{out,k} - t_{in,k})$, with $t_{out,k} = t_{in,k+1}$ and $t_{in,0} = t_0$, so that during each time interval $[t_{in,k},t_{out,k}]$ the trajectory stays in a given region. Then, replacing $t_a$ by $t_{in,k}$ and $t_b$ by $t_{out,k}$ in~\eqref{eq:intSij_dL2S} and~\eqref{eq:intSii_dL2S}, adding up to $q$ for every region $X_\ij$ traversed, and using the continuity of $S$ yields
	\begin{equation}
	\label{eq:intS_dL2S}
		S(x(t_1),\tx(t_1)) - S(x(t_0),\tx(t_0)) + \int_{t_0}^{t_1}\! \norm{\ytau \!-\! \tytau}^2d\tau - \eta^2\!\!\int_{t_0}^{t_1}\! \norm{\utau \!-\! \tutau}^2\,d\tau \leq 0
	\end{equation}
	 
	From~\eqref{eq:DI}, this shows that $S$ is a storage function such that the augmented system $\bSpwa$ is dissipative with respect to the supply rate~\eqref{eq:IncSupRate}. Theorem~\ref{th:IncDiss} thus implies that $\Spwa$ has an \dLdg less than or equal to $\eta$, which concludes the proof.
\end{proof}

\begin{remark}
\label{rem:Remarks}
~\vspace{-2em}
	\begin{enumerate}[(i)]
		\item Conditions~\eqref{eq:ContThm1_dL2S} and~\eqref{eq:ContThm2_dL2S} are non-strict LMIs that, alongside the matrix equality~\eqref{eq:ContThm3}, can be efficiently handled by semi-definite programming solvers. \label{rem:LMIeff}
		\item It can be shown that condition~\eqref{eq:ContThm3} is also necessary for continuity. \label{rem:ContNec}
		\item The three terms $\bG_{ij}^TU_\ij\bG_{ij}$, $\bG_{ij}^TR_\ij\bG_{ij}$ and $\bG_{ij}^TW_\ij\bG_{ij}$ in~\eqref{eq:ContThm2_dL2S},~\eqref{eq:ContThm2_dAS} and~\eqref{eq:ContThm2_dL2S+dAS} come from the application of the \Sproc~\cite{Johansson1998}. Since the \Sproc is not lossless in general, some conservativeness is obtained. We may potentially reduce these effects by considering also other \Sproc terms, e.g. as in~\cite{Hassibi1998}. \label{rem:Sproc}
	\end{enumerate}
\end{remark}

\subsection{Piecewise quadratic incremental exponential stability}
\label{ss:PWQdES}

We provide in this section sufficient conditions for incremental asymptotic stability of PWA systems. We recall that, according to Definition~\ref{def:dAS}, this property is concerned with the convergence of any pair of trajectories driven by the same input $u \in \U$. In this sense, using the fact that $u = \tu$, the augmented system~\eqref{eq:PWAsysUYaug} may be rewritten as

\begin{equation}
\label{eq:PWAsysUYaugAS}
	\by = \bSpwa(u) \left\lbrace \begin{aligned} & \begin{aligned} \dbxt &= \bA_\ij\bxt + \bF_\ij\ut \\ \byt &= \bC_\ij\bxt \end{aligned} & \text{for } \bxt\in X_\ij \\
		& \bx(0) = \bx_0
	\end{aligned}\right.
\end{equation}
with $\bF_\ij$ given by
\begin{equation}
	\bF_\ij = \nmatrix{c}{B_i \\ B_j \\ 0}
\end{equation}

We are now able to state the following theorem.

\begin{theorem}
\label{th:PWAdAS}
	If there exist symmetric matrices $P_{i} \in \R^{n \times n}$ and $\bP_{ij} \in \R^{(2n+1) \times (2n+1)}$; $U_{ij}$, $R_{ij}$, $W_{ij} \in \R^{p_{ij} \times p_{ij}}$ with nonnegative coefficients and zero diagonal;  $L_{ik} \in \R^{(2n+1)\times 1}$ and $\sigma_1, \sigma_2, \sigma_3$ strictly positive such that
	\begin{align}
			& \begin{cases}
				P_{i} - \sigma_1I_n \succeq 0 \\
				P_{i} - \sigma_2I_n \preceq 0 \\
				A_i^TP_{i} + P_{i}A_i + \sigma_3I_n \preceq 0
			\end{cases} & &\text{for } i \in \I \label{eq:ContThm1_dAS}
			\\
			& \begin{cases}
				\bP_{ij} - \sigma_1\bJ_n - \bG_{ij}^TU_{ij}\bG_{ij} \succeq 0 \\
				\bP_{ij} - \sigma_2\bJ_n + \bG_{ij}^TR_{ij}\bG_{ij} \preceq 0 \\
				\bA_\ij^T\overline{P}_\ij + \overline{P}_\ij \bA_\ij + \sigma_3\bJ_n + \bG_{ij}^TW_{ij}\bG_{ij} \preceq 0 \\
				\bP_\ij\bF_\ij = 0
			\end{cases} & &\text{for } (i,j) \in \I \times \I,\, i \neq j\label{eq:ContThm2_dAS}
	\end{align}
	 and~\eqref{eq:ContThm3} are satisfied, then the piecewise-affine system~\eqref{eq:PWAsysUY} is incrementally exponentially stable on $X$ with respect to $\LdepRp$.
\end{theorem}

\begin{proof}
	We shall demonstrate that the above conditions allow us to build a continuous piecewise-quadratic incremental Lyapunov function $V$, given by the same structure as $S$ in~\eqref{eq:PWQStorFunc}, which is shown to be bounded by quadratic functions and to decrease exponentially. This allows us to prove incremental exponential stability of~\eqref{eq:PWAsysUY}.
	
	\noindent\emph{Continuity} -
	Follows exactly as in Theorem~\ref{th:PWAdL2S}.
	
	\noindent\emph{Norm bounds} -
	The first inequality in~\eqref{eq:ContThm2_dAS}, post and pre multiplied respectively by $\bx$ and $\bx^T$, implies that $\bx^T\bP_\ij\bx - \sigma_1\norm{x - \tx}^2 \geq \bx^T\bG_{ij}^TU_{ij}\bG_{ij}\bx$. Since $U_{ij}$ is composed of nonnegative coefficients, the right-hand side of the previous inequality is nonnegative whenever $\bx \in X_\ij$. This implies that
	\begin{equation}
	\label{eq:PWQPosDefXij_dAS}
	 	\bx^T\bP_\ij\bx  \geq \sigma_1\norm{x - \tx}^2 \qquad \text{for } \bx \in X_{ij}
	\end{equation}
	
	The first inequality in~\eqref{eq:ContThm1_dAS} implies that $V(x,\tx) \geq \sigma_1\norm{x - \tx}^2$ for all $\bx \in X_{ii}$. With~\eqref{eq:PWQPosDefXij_dAS}, this guarantees that	
	\begin{equation}
	\label{eq:PWQPosDef_dAS}
		V(x,\tx) \geq \sigma_1\norm{x - \tx}^2, \quad\forall x,\tx \in X
	\end{equation}
	
	Proceeding exactly as before, the second inequalities in~\eqref{eq:ContThm1_dAS} and~\eqref{eq:ContThm2_dAS} imply that
	\begin{equation}
	\label{eq:PWQUpBound_dAS}
		V(x,\tx) \leq \sigma_2\norm{x - \tx}^2, \quad\forall x,\tx \in X
	\end{equation}
	Inequalities~\eqref{eq:PWQPosDef_dAS} and~\eqref{eq:PWQUpBound_dAS} imply that the continuous piecewise-quadratic function $V$ is such that
	\begin{equation}
	\label{eq:PWQBounds_dAS}
	 	\sigma_1\norm{x - \tx}^2 \leq V(x,\tx) \leq \sigma_2\norm{x - \tx}^2
	\end{equation}
		
	\noindent\emph{Exponential decay} -
	We now show that the function $V$ decays exponentially and conclude on the incremental exponential stability. Using the same arguments as before, the third inequality in~\eqref{eq:ContThm2_dAS}, post and pre multiplied by $\bx^T$ and $\bx$, implies that
	\begin{equation}
		\bx^T\bP_\ij\bA_\ij\bx + \bx^T\bA_\ij^T\bP_\ij\bx \leq -\sigma_3\norm{x - \tx}^2
	\end{equation}
	for all $\bx \in X_\ij$. From the equality in~\eqref{eq:ContThm2_dAS}, and using the fact that $\dbx = \bA_\ij\bx + \bF_\ij u$, we may write
	\begin{equation}
		\bx^T\bP_\ij\dbx + \dbx^T\bP_\ij\bx \leq -\sigma_3\norm{x - \tx}^2
	\end{equation}
	In the interior of each region $\Xij$, $V$ is differentiable and such that $\dV(x,\tx)$ is equal to the left-hand side of the previous inequality. Using this fact together with~\eqref{eq:PWQUpBound_dAS}, we obtain
	\begin{equation}
		\dV(x,\tx) + \frac{\sigma_3}{\sigma_2}V(x,\tx) \leq 0
	\end{equation}
	Let $t_a$ and $t_b$ be two time instants such that the state trajectory of system~\eqref{eq:PWAsysUYaugAS} remains in $\Xij$ on the interval $[t_a,t_b]$. By noticing that $\dbx = \bA_\ij\bx + \bF_\ij u$, and integrating from $t_a$ to $t_b$ along trajectories of~\eqref{eq:PWAsysUYaugAS}, we have
	\begin{equation}
	\label{eq:intSij_dAS}
		V(x(t_b),\tx(t_b))e^{(\sigma_3/\sigma_2)t_b} - V(x(t_a),\tx(t_a))e^{(\sigma_3/\sigma_2)t_a} \leq 0
	\end{equation}
	
	The same reasoning can be applied to the last inequality in~\eqref{eq:ContThm1_dAS}. Let us consider two trajectories $\xt = \phi(t,0,\xo,u)$ and $\txt = \phi(t,0,\txo,u)$, for $u \in \LdepRp$. The time $t$ can be decomposed as $t = t - t_{in,q} + \sum_{k=0}^{q-1} (t_{out,k} - t_{in,k})$, with $t_{out,k} = t_{in,k+1}$ and $t_{in,0} = 0$, so that during each time interval $[t_{in,k},t_{out,k}]$ the trajectory stays in a given region. Then, replacing $t_a$ by $t_{in,k}$ and $t_b$ by $t_{out,k}$ in~\eqref{eq:intSij_dAS}, adding up to $q$ for every region $X_\ij$ crossed, and using continuity yields
	\begin{equation}
	\label{eq:intS_dAS}
		V(\xt,\txt)e^{(\sigma_3/\sigma_2)t} - V(\xo,\txo) \leq 0
	\end{equation}
	and then, from~\eqref{eq:PWQBounds_dAS},
	\begin{equation}
	 	\norm{\xt - \txt} \leq \sqrt{\frac{\sigma_2}{\sigma_1}}e^{-(\sigma_3/2\sigma_2)t}\norm{x_0 - \tx_0}
	\end{equation}
	 which concludes the proof.
\end{proof}

\begin{remark}
	We remark that the ratio $\sigma_3/2\sigma_2$ appearing in the previous proof constitutes an upper bound on the incremental exponential decay rate. This might be taken into account when solving the LMIs in Theorem~\ref{th:PWAdAS}, by maximizing $\sigma_3$ or minimizing $\sigma_2$, in order to compute a less conservative bound. Evidently, it would be possible to compute a direct bound on the decay rate, as it is standard procedure on Lyapunov analysis (see e.g.~\cite{Boyd1994}).
\end{remark}

\begin{remark}
	The function $V$ constructed in Theorem~\ref{th:PWAdAS} possesses stronger properties than the function in Theorem~\ref{th:dAS} due to the quadratic bounds provided by the $\sigma_1, \sigma_2$ and $\sigma_3$ terms, hence allowing us to derive incremental exponential stability, while asymptotic stability is obtained in~\cite{Angeli2002}.
\end{remark}

\subsection{A single theorem for characterizing both input-output and internal incremental stability}
\label{ss:PWQdL2S+dAS}

This section presents a third theorem ensuring both incremental \Ldg and incremental exponential stability of PWA systems. 

\begin{theorem}
\label{th:PWAdL2S+dAS}
	If there exist symmetric matrices $P_{i} \in \R^{n \times n}$ and $\bP_{ij} \in \R^{(2n+1) \times (2n+1)}$; $U_{ij}$, $R_{ij}$, $W_{ij} \in \R^{p_{ij} \times p_{ij}}$ with nonnegative coefficients and zero diagonal; $L_{ijkl} \in \R^{(2n+1)\times 1}$ and $\sigma_1, \sigma_2, \sigma_3 > 0$ such that
	\begin{align}
		& \begin{cases}
			P_{i} - \sigma_1I_n \succeq 0 \\
			P_{i} - \sigma_2I_n \preceq 0 \\
			\nmatrix{cc}{A_i^TP_{i} + P_{i}A_i + C_i^TC_i + \sigma_3I_n & P_iB_i + C_i^TD \\ \bullet & D^TD - \eta^2I_p} \preceq 0
		\end{cases} & &\text{for } i \in \I \label{eq:ContThm1_dL2S+dAS}\\
		& \begin{cases}
			\bP_{ij} - \sigma_1\bJ_n - \bG_{ij}^TU_{ij}\bG_{ij} \succeq 0 \\
			\bP_{ij} - \sigma_2\bJ_n + \bG_{ij}^TR_{ij}\bG_{ij} \preceq 0 \\
			\nmatrix{cc}{
				\left({ \begin{smallmatrix} \textstyle
      				\bA_\ij^T\overline{P}_\ij + \overline{P}_\ij \bA_\ij + \bC_\ij^T\bC_\ij + {} \\ \textstyle\sigma_3\bJ_n + \bG_{ij}^TW_{ij}\bG_{ij}
    			\end{smallmatrix}}\right) &
				\bP_\ij\bB_\ij + \bC_\ij^T\bD \\
				\bullet &
				\bD^T\bD - \eta^2\overline{I}_p} \preceq 0
		\end{cases}  & &\begin{aligned} \text{for } (i,j) &\in \I \times \I \\ i &\neq j \end{aligned}\label{eq:ContThm2_dL2S+dAS}
	\end{align}
	 and~\eqref{eq:ContThm3} are satisfied, then the piecewise-affine system~\eqref{eq:PWAsysUY} is incrementally \Ldg stable with incremental \Ldg less than or equal to $\eta$ and is incrementally asymptotically stable with respect to $\U = \LdpRp$. Furthermore, $S$ given as in~\eqref{eq:PWQStorFunc} is both a storage function for the augmented system~\eqref{eq:PWAsysUYaug} and an incremental Lyapunov function.
\end{theorem}

\begin{proof}
	It is clear that feasibility of~\eqref{eq:ContThm1_dL2S+dAS} implies the feasibility of~\eqref{eq:ContThm1_dL2S} and~\eqref{eq:ContThm1_dAS}, and feasibility of~\eqref{eq:ContThm2_dL2S+dAS} implies the feasibility of~\eqref{eq:ContThm2_dL2S} and the inequalities in~\eqref{eq:ContThm2_dAS}. We note that the equality in~\eqref{eq:ContThm2_dAS} is also implied by the last inequality in~\eqref{eq:ContThm2_dL2S}. Indeed, using the change of variables
	\begin{equation}
	\label{eq:ChVarU-DU}
	 	\nmatrix{c}{u \\ \tu} = \frac{1}{2}\nmatrix{cc}{I_p & I_p \\ -I_p & I_p}\nmatrix{c}{u - \tu \\ u + \tu}
	\end{equation}
	the lower right block of~\eqref{eq:ContThm2_dL2S} becomes
	\begin{equation}
	\label{eq:LowRightBlock}
	 	\nmatrix{cc}{D^TD - \eta^2I_p & 0 \\ 0 & 0}
	\end{equation}
	and the upper right block becomes
	\begin{equation}
	\label{eq:UppRightBlock}
	 	\frac{1}{2}\bPij\nmatrix{cc}{B_i & B_i \\ -B_j & B_j \\ 0 & 0} + \bCij^T\nmatrix{cc}{D & 0}
	\end{equation}
	The presence of zero values at the diagonal in~\eqref{eq:LowRightBlock} implies that all values in the corresponding rows and columns must be zero for the LMI to be feasible. Hence, from~\eqref{eq:UppRightBlock} we infer that $\bPij\bFij = 0$. Thus, feasibility of the LMIs~\eqref{eq:ContThm1_dL2S+dAS},~\eqref{eq:ContThm2_dL2S+dAS} and the matrix equality~\eqref{eq:ContThm3} allows us to construct a function $S$ given by~\eqref{eq:PWQStorFunc}, which is both a storage function and an incremental Lyapunov function, proving incremental \Ldg stability and incremental exponential stability through the same arguments as in Theorems~\ref{th:PWAdL2S} and~\ref{th:PWAdAS}.
\end{proof}

\begin{remark}
	We note that inequalities~\eqref{eq:ContThm2_dL2S} and~\eqref{eq:ContThm2_dL2S+dAS} would not be feasible for systems with matrix $D$ depending on the regional partition. Indeed, suppose matrix $\bD$ of the augmented system~\eqref{eq:PWAsysUYaug} was replaced by
	\begin{equation}
	 	\bDij = \nmatrix{cc}{D_i & -D_j}
	\end{equation}
	In this way, the lower right block of~\eqref{eq:ContThm2_dL2S} and~\eqref{eq:ContThm2_dL2S+dAS} would become $\bDij^T\bDij - \eta^2\bI_p$. Using again the change of variables~\eqref{eq:ChVarU-DU} on this block yields the matrix
	\begin{equation}
	 	\frac{1}{4}\nmatrix{cc}{(D_i + D_j)^T(D_i + D_j) - \eta^2I_p & (D_i + D_j)^T(D_i - D_j) \\ (D_i - D_j)^T(D_i + D_j) & (D_i - D_j)^T(D_i - D_j)}
	\end{equation}
	The lower right diagonal block must be negative semidefinite for the inequalities to be feasible, and hence $D_i = D_j$.
\end{remark}

\begin{remark}
\label{rem:LowBound}
	Inequalities~\eqref{eq:ContThm1_dAS} and~\eqref{eq:ContThm1_dL2S+dAS} are not feasible for $A_i$ non Hurwitz, and hence Theorems~\ref{th:PWAdAS} and~\ref{th:PWAdL2S+dAS} require that each subsystem be asymptotically stable. Additionally, for Theorem~\ref{th:PWAdL2S+dAS},~\eqref{eq:ContThm1_dL2S+dAS} requires that the \Hoo norm of each subsystem be less than or equal to $\eta$.
\end{remark}

In this section we have proposed conditions that allow assessment of \dLdg and incremental asymptotic stability of PWA systems. Despite the fact that the LMIs in Theorems~\ref{th:PWAdL2S} and~\ref{th:PWAdAS} are quite similar, they are not equivalent. In Section~\ref{se:ObsReachCon}, it will be established how observability and reachability can be used to bridge the gap between both notions in the case of general nonlinear systems $\Sigma$. However, as pointed out in~\cite{Bemporad2000}, observability and reachability are complex properties for PWA systems, and cannot be simply inherited from the respective properties of the subsystems. Theorem~\ref{th:PWAdL2S+dAS} allows us to avoid this problem by establishing \dLdg stability and incremental asymptotic stability concurrently.

\section{Numerical examples}
\label{se:NumEx}

In this section we consider some examples of incremental analysis of PWA systems. These examples illustrate the use of the conditions obtained in the last section. Example~\ref{ex:Egg} shows how piecewise-quadratic storage functions allow us to compute an upper bound where quadratic functions fail. Example~\ref{ex:DeadZone} illustrates how the \dLdg may be significantly greater than the \Ldg of a system, and justifies the interest in addressing the former. Finally, in Example~\ref{ex:Sat}, the continuous piecewise-quadratic structure is highlighted through a contour plot of the storage function.

\begin{example}
\label{ex:Egg}
	Consider the PWA system given by~\eqref{eq:PWAsysUY} with
	\begin{equation}
	\begin{aligned}
		A_1 &= \nmatrix{cc}{-0.1 & 1\\ -5 & -0.1} & A_2 &= \nmatrix{cc}{-0.1 & 1\\ -1 & -0.1} 
	\end{aligned}
	\end{equation}
	$a_1 = a_2 = 0$, $B_1 = B_2 = \nmatrix{cc}{0 & 1}^T$, $C_1 = C_2 = \nmatrix{cc}{1 & 0}$, $c_1 = c_2 = 0$ and $D = 0$. The state space partition is illustrated in Fig.~\ref{fig:EggTraj}, along with a sample trajectory for $u \equiv 0$.
		\begin{figure}[tb]
		\centering
		\begin{subfigure}[t]{0.4\textwidth}
			\centering
			\includegraphics{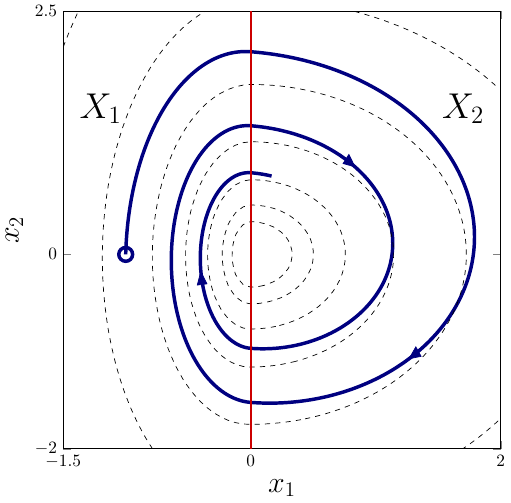}	
			\caption{A sample trajectory for $u \equiv 0$ illustrating the state space partitioning}
			\label{fig:EggTraj}
		\end{subfigure}
		\quad
		\begin{subfigure}[t]{0.44\textwidth}
			\centering
			\includegraphics{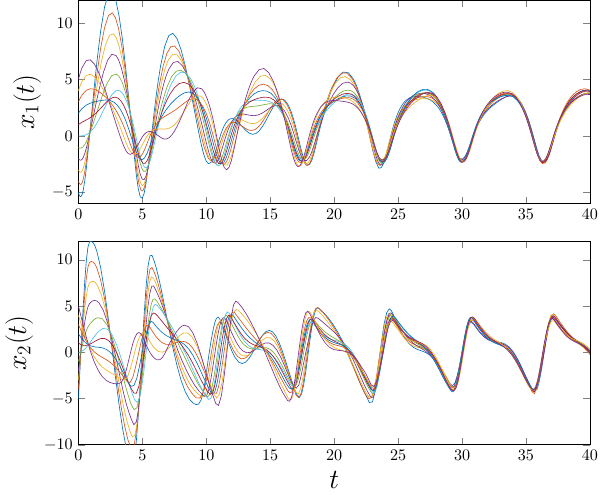}	
			\caption{Response for a sinusoidal input and different initial conditions}
			\label{fig:EggSync}
		\end{subfigure}
		\caption{Trajectories of the system in Example~\ref{ex:Egg}}
	\end{figure}
	Fig.~\ref{fig:EggSync} presents some trajectories of $\Spwa$ for different initial conditions and a sinusoidal input, where we can see that they all converge to one another in a periodic steady state, which suggests that the system is incrementally asymptotically stable. We can show that the approach in~\cite{Romanchuk1999} is not able to find a quadratic storage function for this system. The conditions in~\cite{Romanchuk1999} give a quadratic storage function that, when found, is also an incremental Lyapunov function. Then, based on Lemma 5.1 from~\cite{Angeli2009}, this system must admit a quadratic Lyapunov function. Since we can show that no quadratic Lyapunov function exists\footnote{The search for a quadratic Lyapunov function can be expressed as a set of LMIs whose feasibility, in this case, is a necessary and sufficient condition for quadratic stability, and for which no solution can be found.}, no quadratic storage function may exist. Using Theorem~\ref{th:PWAdL2S+dAS}, a piecewise-quadratic storage function may be found and we compute an upper bound on the \dLdg of $\eta = 5.005$. Additionally, Theorem~\ref{th:PWAdL2S+dAS} also ensures that the system is incrementally asymptotically stable. This example illustrates how the search for a piecewise-quadratic storage function is less conservative than a single quadratic function, and allows us to conclude where the latter fails.
\end{example}

\begin{example}
\label{ex:DeadZone}
	Consider the scalar system $\dx = -\kappa(x) + u$ with the output $y = x$. The function $\kappa$ is given by
	\begin{equation}
		\kappa(x) = \begin{cases}
			x  & |x| \leq 1 \\
			\frac{1}{10}x + \frac{9}{10}\sign(x) & 1 < |x| \leq \frac{9}{4} \\
			x - \frac{9}{8}\sign(x) & |x| > \frac{9}{4}
		\end{cases}
	\end{equation}
	This system admits a PWA representation given by
	\begin{equation}
	\begin{aligned}
		A_1 &= -1 \!&\! A_2 &= -\frac 1 {10} \!&\! A_3 &= -1 \!&\! A_4 &= -\frac 1 {10} \!&\! A_5 &= -1 \\
		a_1 &= -\frac 9 8 \!&\! a_2 &= \frac 9 {10} \!&\! a_3 &= 0 \!&\! a_4 &= -\frac 9 {10} \!&\! a_5 &= \frac 9 8
	\end{aligned}
	\rule[-2.5em]{0pt}{0pt}
	\end{equation}
	and $B_i = C_i = 1$, $c_i = 0$, for all $i \in \I$. We shall analyze the difference between the \Ldg and the \dLdg of this system. Using the techniques described e.g. in~\cite{Johansson2003} to find an upper bound on the \Ldg of this system yields $\gamma = 2$. Theorem~\ref{th:PWAdL2S} may be used to estimate an upper bound on the incremental \Ldg, which yields $\eta = 10$. Using the techniques described e.g. in~\cite{Johansson2003} to find an upper bound on the \Ldg of this system yields $\gamma = 2$. A simulation of this system using \Matlab with $\ut = b$ and $\tut = A\sin{(\omega_0t)} + b$ for $A = 0.05$, $b = 1.05$ and $\omega_0 = 0.05$ rad/s from $t = 0$ to $T = 100$ s yields
	\begin{equation}
	 	\frac{\int_{0}^{T}\! \norm{y(\tau) - \ty(\tau)}^2\,d\tau}{\int_{0}^{T}\! \norm{u(\tau) - \tu(\tau)}^2\,d\tau} \approx 8.9
	\end{equation}
	which gives a lower bound on the incremental \Ldg. This simple example illustrates how the \dLdg may be significantly greater than the \Ldg. This shows that \Ldg stability and \dLdg stability are related but different concepts, with the latter being a stronger property.

\end{example}

\begin{example}
\label{ex:Sat}
	
	Let us consider the linear system described by the transfer function $H(s) = (s + 3)/(s + 1)$ that is negatively fed back with a saturated linear gain $\sigma$ given by
	\begin{equation}
	 	\sigma(y) = \begin{cases}
	 		h\sign(y) & |y| > \frac{h}{k} \\
	 		ky & |y| \leq \frac{h}{k}
	 	\end{cases}
	\end{equation}
	The closed loop system admits a PWA representation given by
	\begin{equation}
	\begin{aligned}
		A_1 &= -1 & A_2 &= -\frac{3k+1}{k+1} & A_3 &= -1\\
		a_1 &= -2h & a_2 &= 0 & a_3 &= 2h
	\end{aligned}
	\end{equation}
	and $B_i = 2$, $C_i = 1$, $c_i = 0$, for all $i \in \I$, and $D = 1$. For $h = 5$ and $k = 1$, applying Theorem~\ref{th:PWAdL2S} one can find a continuous piecewise-quadratic storage function $S$ that ensures global incremental asymptotic stability. Fig.~\ref{fig:SatContour} presents the contour plot of $S$, where we can see it is indeed a piecewise-quadratic function of $\bx$. This highlights how the storage function~\eqref{eq:PWQStorFunc} is more flexible than a quadratic function, and thus the results obtained are potentially less conservative. 
	\begin{figure}[tb]
		\centering
		\includegraphics{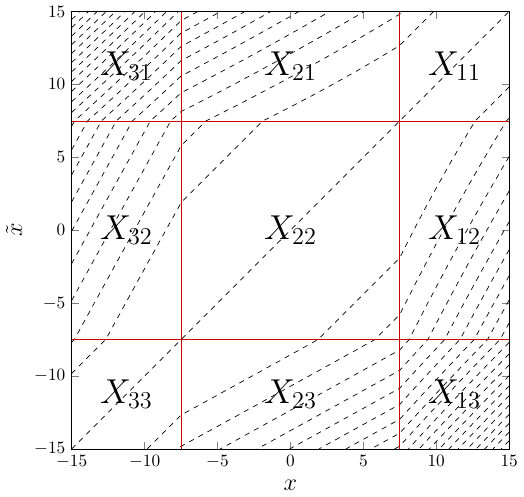}
		\caption{Contour plot of the storage function for the augmented system in Example~\ref{ex:Sat} illustrating its PWA structure}
		\label{fig:SatContour}
	\end{figure}
\end{example}

\section{Connections between \dLdg stability and incremental asymptotic stability}
\label{se:ObsReachCon}

Section~\ref{se:AnalPWA} presents sufficient conditions to assess \dLdg stability and incremental asymptotic stability of nonlinear systems given by a PWA representation. Albeit different, these two notions are fundamentally connected, as we aim to demonstrate in this section. The developments of this section are not restricted to PWA systems, but are valid for general autonomous nonlinear systems with a state representation given by~\eqref{eq:DynSys}. In view of the different nature of both concepts, one being defined as an input-output property and the other characterized in state space, it is clear that we shall need relevant notions of observability and reachability, as is the case when dealing with classical notions of finite gain and asymptotic stability~\cite{Willems1971b,Hill1980a}.

\subsection{Observability and reachability notions}
\label{ss:ObsReach}

We begin by defining a suitable notion of observability. The following definition is based on the notions of uniform irreducibility~\cite{Willems1971b} and incremental observability~\cite{Bemporad2000}.

\begin{definition}[Uniform and quadratic observability]
\label{def:Obs}
	$\Sigma$ is said to be uniformly observable on $\X \subset X$ with respect to $\U \subset \LdpRp$ if there exist $\alpha_o \in \Koo$, and a constant $T_o \geq 0$ such that
	\begin{equation}
	\label{eq:ObsIneq}
	 	\int_t^{t+T_o}\! \norm{y(\phi(\tau,t,x,u_0),u_0(\tau)) - y(\phi(\tau,t,\tx,u_0),u_0(\tau))}^2 \,d\tau \geq \alpha_o(\norm{x - \tx})
	\end{equation}
	for all $x, \tx \in \X$, $t \geq 0$ and any $u_0 \in \U$. If $\alpha_o$ is a quadratic function $\alpha_o(r) = \sigma_o\norm{r}^2$, with $\sigma_o > 0$, then the system is said to be quadratically observable on $\X$ with respect to $\U$.
\end{definition}

We should note that the main difference between the previous definition and the one presented in~\cite{Willems1971b} lies in the fact that the property is defined for any $u \in \U$, instead of for some $u$. This distinction is important as it will allow to obtain incremental asymptotic stability of~\eqref{eq:DynSys} for every $u \in \U$. We state now a definition of reachability taken from~\cite{Fromion1997}, for our purposes renamed uniform reachability.

\begin{definition}[Uniform and quadratic reachability]
\label{def:Reach}
	The state space of $\Sigma$ is said to be uniformly reachable from $x_0$ if it is reachable from $\xo$ and there exist $\alpha_r \in \Koo$ and $T_r \geq 0$ satisfying
	\begin{equation}
	\label{eq:ReachIneq}
	 	\int_{t-T_r}^t\! \norm{u(\tau) - \tu(\tau)}^2 \, d\tau \leq \alpha_r(\norm{x - \tx})
	\end{equation}
	for all $x, \tx \in X$ and $t \geq T_r$, where $u, \tu \in \LdpRp$ and $x = \phi(t,t-T_r,x_0,u)$, $\tx = \phi(t,t-T_r,x_0,\tu)$. If $\alpha_r$ is a quadratic function $\alpha_r(r) = \sigma_r\norm{r}^2$, with $\sigma_r > 0$, then the system is said to be quadratically reachable from $\xo$.
\end{definition}

As it is standard procedure when dealing with non-incremental properties (see e.g.~\cite{Vidyasagar1993}), the stronger assumption requiring $\alpha_o$ and $\alpha_r$ to be quadratic functions will be capital in establishing incremental exponential stability.

\subsection{Intermediary results}
\label{ss:Lemmas}

We shall make further use of the dissipativity framework, which is a powerful tool to link input-output properties and state-space behavior. Before stating some key lemmas that will be used in the follow-up, let us recall some fundamental definitions and results concerning dissipativity, most of which are taken or adapted from the seminal paper~\cite{Willems1972a}. We begin by defining the available storage function.

\begin{definition}[Available storage]
	The available storage of system~\eqref{eq:DynSysAug} with supply rate $w$ is the function from $X \times X$ to $\bR_+$ defined by
	\begin{multline}
	\label{eq:SaDef}
	 	S_a\big(x,\tx\big) = \sup \left\lbrace-\int_t^{t+T}\! w\big(\utau,\tutau,\bytau\big) \,d\tau \Biggm| T \geq 0, (u,\tu,x,\tx,\by) \text{ satisfy}~\eqref{eq:DynSysAug}\right. \\
	 	\left.\text{with } \xt = x, \txt = \tx, \text{ and } u,\tu \in \LdpRp \vphantom{\Biggm|}\right\rbrace
	\end{multline}
	The supremum is taken over all $T \geq 0$, and all motions starting in $(x,\tx)$ at time $t = 0$ under any $u,\tu \in \LdpRp$.
\end{definition}

The following theorem, adapted from~\cite{Willems1972a}, gives an important characterization of the available storage function.

\begin{theorem}
	The available storage $S_a$ of system~\eqref{eq:DynSysAug} with respect to~\eqref{eq:IncSupRate}, is finite for all $x,\tx \in X$ if and only if $\Sigma$ is incrementally dissipative. Moreover, any storage function $S$ is such that $0 \leq S_a \leq S$ for dissipative dynamical systems and $S_a$ is itself a possible storage function.
\end{theorem}

We now state a series of lemmas concerning the implications of observability and reachability over the structure of the storage function $S$. These results will be useful in establishing Theorem~\ref{th:dL2+Obs/Reach->dAS}. We begin by showing that the storage function is non-increasing along motions driven by the same input $u = \tu$.

\begin{lemma}
\label{lem:SNonIncreas}
	If system~\eqref{eq:DynSys} is incrementally dissipative, then every storage function $S$ is such that
	\begin{equation}
	 	S\big(x(t_2),\tx(t_2)\big) \leq S\big(x(t_1),\tx(t_1)\big)
	\end{equation}
	for any $t_2 \geq t_1$, with $x(t_2) = \phi(t_2,t_1,x(t_1),u_0)$ and $\tx(t_2) = \phi(t_2,t_1,\tx(t_1),u_0)$ for any $u_0 \in \LdpRp$.
\end{lemma}

\begin{proof}
	From the dissipation inequality~\eqref{eq:DI} from $t_1$ to $t_2$ we have
	\begin{equation}
	 	S\big(x(t_2),\tx(t_2)\big) - S\big(x(t_1),\tx(t_1)\big) \leq \eta^2\int_{t_1}^{t_2}\! \norm{u(\tau) - \tu(\tau)}^2\,d\tau - \int_{t_1}^{t_2}\!\norm{y(\tau) - \ty(\tau)}^2 \, d\tau
	\end{equation}
	and hence choosing $u = \tu \in \LdpRp$ yields
	\begin{equation}
	 	S\big(x(t_2),\tx(t_2)\big) - S\big(x(t_1),\tx(t_1)\big) \leq  - \int_{t_1}^{t_2}\!\norm{y(\tau) - \ty(\tau)}^2 \, d\tau \leq 0
	\end{equation}
	which proves the statement.
\end{proof}

The next two lemmas are adapted from~\cite{Fromion1997}. We begin by showing that uniform observability and reachability imply that the storage function is upper and lower bounded by class $\Koo$ functions of $\norm{x - \tx}$.

\begin{lemma}
\label{lem:SBounds}
	If system~\eqref{eq:DynSys} is incrementally \Ldg stable with an \dLdg less than or equal to $\eta$, and it is uniformly observable on $\X$ with respect to $\U$ and uniformly reachable from $\xo$, then there exist $\alpha_o, \alpha_r \in \Koo$ and a storage function $S$ such that
	\begin{equation}
	\label{eq:SBounds}
 		\alpha_o(\norm{x - \tx}) \leq S(x,\tx) \leq \eta^2\alpha_r(\norm{x - \tx})
	\end{equation}
	on the set $\X$. If in addition~\eqref{eq:DynSys} is quadratically observable on $\X$ with respect to $\U$ and quadratically reachable from $\xo$, then there exist positive scalars $\sigma_o$ and $\sigma_r$ such that
	\begin{equation}
	\label{eq:SQBounds}
		\sigma_o \norm{x - \tx}^2 \leq S(x,\tx) \leq \eta^2\sigma_r\norm{x - \tx}^2
	\end{equation}
\end{lemma}
\begin{proof}	
	\noindent\emph{Lower bound} -
	Let us consider the available storage function~\eqref{eq:SaDef}. The suboptimality of the couple of inputs $u = \tu \in \U$ yields
	\begin{equation}
	 	S_a(\xt,\txt) \geq \int_t^{t+T}\!\norm{y(\tau) - \ty(\tau)}^2 \, d\tau \quad \forall T \geq 0
	\end{equation}
	where $\ytau = h(\phi(\tau,t,\xt,u),\utau)$ and $\tytau = h(\phi(\tau,t,\txt,u),\utau)$. Since the system is uniformly observable on $\X$ with respect to $\U$, there exist $T_o \geq 0$ and $\alpha_o \in \Koo$ such that~\eqref{eq:ObsIneq} is respected. Picking $T \geq T_o$ in the previous inequality, and using the fact that $S(x,\tx) \geq S_a(x,\tx)$ yields
	\begin{equation}
	 	S(x,\tx) \geq \alpha_o(\norm{x - \tx}), \qquad \forall x, \tx \in \X
	\end{equation}
	
	\noindent\emph{Upper bound} -
	We recall the dissipation inequality~\eqref{eq:DI} taken from $t_0 = 0$:
	\begin{equation}
	 	S\big(\xt,\txt\big) - S\big(x_0,\tx_0\big) \leq \eta^2\int_{0}^{t}\! \norm{u(\tau) - \tu(\tau)}^2\,d\tau - \int_{0}^{t}\!\norm{y(\tau) - \ty(\tau)}^2 \, d\tau
	\end{equation}
	By choosing $x_0 = \tx_0$ and recalling that $S(\xo,\xo) = 0$, $\forall \xo \in X$, we can write
	\begin{equation}
	 	S\big(x(t),\tx(t)\big) \leq \eta^2\int_{0}^{t}\! \norm{u(\tau) - \tu(\tau)}^2\,d\tau
	\end{equation}
	Since the system is uniformly reachable from $\xo$, there exist $T_r \geq 0$ and $\alpha_r \in \Koo$ such that~\eqref{eq:ReachIneq} is respected. For $t \geq T_r$, and choosing $u, \tu$ such that $u(\tau) = \tu(\tau)$ for $0 \leq \tau < t - T$, the previous inequality becomes
	\begin{equation}
	 	S(x,\tx) \leq \eta^2\alpha_r(\norm{x - \tx})
	\end{equation}
	as claimed. The quadratic versions follow in the same manner.
\end{proof}

The next lemma gives an important characterization of the decrease of the storage function along observable motions.

\begin{lemma}
\label{lem:S-SNegDef}
	If system~\eqref{eq:DynSys} is incrementally \Ldg stable and uniformly observable on $\X$ with respect to $\U$, then there exist $T_o \geq 0$ and $\alpha_o \in \Koo$ so that, for every $T \geq T_0$, all storage functions $S$ respecting the dissipation inequality~\eqref{eq:DI} with the supply rate given by~\eqref{eq:IncSupRate} are such that
	\begin{equation}
	\label{eq:S-SNegDef}
	 	S\big(x(t+T),\tx(t+T)\big) - S\big(\xt,\txt\big) \leq -\alpha_o(\norm{\xt - \txt})
	\end{equation}
	where $x(t+T) = \phi(t+T,t,\xt,u_0)$ and $\tx(t+T) = \phi(t+T,t,\txt,u_0)$, for  $\xt, \txt \in \X$ and for every $u_0 \in \U$. If in addition system~\eqref{eq:DynSys} is quadratically observable on $\X$ with respect to $\U$, then~\eqref{eq:S-SNegDef} becomes
	\begin{equation}
	\label{eq:S-SQNegDef}
		S\big(x(t+T),\tx(t+T)\big) - S\big(\xt,\txt\big) \leq -\sigma_o\norm{x - \tx}^2
	\end{equation}
\end{lemma}
\begin{proof}
	Theorem~\ref{th:IncDiss} states that~\eqref{eq:DynSysAug} is dissipative with respect to~\eqref{eq:IncSupRate}. From the dissipation inequality~\eqref{eq:DI}, considered from $t$ to $t + T$, we have
	\begin{equation}
	 	S\big(x(t+T),\tx(t+T)\big) - S\big(\xt,\txt\big) \leq \eta^2\!\!\int_{t}^{t+T}\! \norm{u(\tau) \!-\! \tu(\tau)}^2\,d\tau - \int_{t}^{t+T}\!\norm{y(\tau) \!-\! \ty(\tau)}^2 \, d\tau
	\end{equation}
	and hence choosing $u, \tu \in \U$ such that $u(\tau) = \tu(\tau)$ for $t \leq \tau \leq t + T$ yields
	\begin{equation}
		S\big(x(t+T),\tx(t+T)\big) - S\big(\xt,\txt\big) \leq - \int_{t}^{t+T}\!\norm{y(\tau) - \ty(\tau)}^2 \, d\tau
	\end{equation}
	Since the system is uniformly observable on $\X$ with respect to $\U$, there exist $T_o \geq 0$ and $\alpha_o \in \Koo$ such that~\eqref{eq:ObsIneq} is respected. Hence for $T \geq T_o$ and $\xt, \txt \in \X$, the claim follows. The quadratic version follows similarly.
\end{proof}

Let us note that, in comparison with Lemma~\ref{lem:SNonIncreas}, observability allows us to quantify the decay of the storage function. Finally, the next lemma gives a different characterization of~\eqref{eq:S-SNegDef} which will be capital in establishing Theorem~\ref{th:dL2+Obs/Reach->dAS}, the main result of this section. This result is adapted from~\cite{Doban2016}.

\begin{lemma}
\label{lem:S-SNegDef2}
If there exists $T_o \geq 0$ and $\alpha_o \in \Koo$ such that the storage function $S$ respects~\eqref{eq:S-SNegDef} for every $T \geq T_o$, then there exists a class $\Koo$ function $\rho$, with $\rho(r) < r$, $\forall r > 0$ such that
\begin{equation}
 	\label{eq:S-SNegDef2}
 	S\big(x(t+T),\tx(t+T)\big) \leq \rho\big(S\big(\xt,\txt\big)\big)
\end{equation}
for every $T \geq T_o$, where $x(t+T) = \phi(t+T,t,\xt,u_0)$ and $\tx(t+T) = \phi(t+T,t,\txt,u_0)$, for  $\xt, \txt \in \X$ and for every $u_0 \in \U$.
\end{lemma}

\begin{proof}
	Let us define $\hat{\alpha}_r := \eta^2\alpha_r$. Using~\eqref{eq:S-SNegDef}, we can write for $x \neq \tx$
	\begin{align}
		0 \leq S\big(x(t+T),\tx(t+T)\big) &\leq S\big(\xt,\txt\big) -\alpha_o(\norm{x - \tx}) \nonumber \\
		&< S\big(\xt,\txt\big) - 0.5\alpha_o(\norm{x - \tx})
	\end{align}
	Using Lemma~\ref{lem:SBounds} then yields
	\begin{equation}
	\label{eq:S-SNegDef2_aux}
		0 \leq S\big(x(t+T),\tx(t+T)\big) < (\id - 0.5\alpha_o\circ\halpha_r^{-1})\big(S\big(\xt,\txt\big)\big)
	\end{equation}
	
	Similarly, for $x \neq \tx$, we have
	\begin{align}
		0 \leq S\big(x(t+T),\tx(t+T)\big) &\leq S\big(\xt,\txt\big) -\alpha_o(\norm{x - \tx}) \nonumber \\
		&< (\halpha_r - 0.5\alpha_o)(\norm{x - \tx})
	\end{align}
	and then $(\halpha_r - 0.5\alpha_o)(r) > 0$, $\forall r > 0$. Since $\halpha_r^{-1} \in \Koo$, we have $(\halpha_r - 0.5\alpha_o)\circ\halpha_r^{-1}(r) > 0$ and then
	\begin{equation}
	 	0 < (\id - 0.5\alpha_o\circ\halpha_r^{-1})(r) < r, \qquad \forall r > 0
	\end{equation}
	Let us write $\hat{\rho} := (\id - 0.5\alpha_o\circ\halpha_r^{-1})$ and note that this function is continuous and positive definite. Since the previous inequality is strict, there exists $\rho \in \Koo$ such that $\hat{\rho}(r) \leq \rho(r) < r$ and then from~\eqref{eq:S-SNegDef2_aux} we can write
	\begin{equation}
		S\big(x(t+T),\tx(t+T)\big) \leq \rho\big(S\big(\xt,\txt\big)\big)
	\end{equation}
	which concludes the proof.
\end{proof}

\subsection{Connection between \dLdg and incremental asymptotic stability}
\label{ss:Connect}

We are now ready to state the following theorem connecting \dLdg stability and incremental asymptotic stability. As stated, the bridge allowing to connect both concepts is based on the appropriate observability and reachability notions.

\begin{theorem}
\label{th:dL2+Obs/Reach->dAS}
	 Assume that system~\eqref{eq:DynSys} is incrementally \Ldg stable with \dLdg less than or equal to $\eta$. Assume also that~\eqref{eq:DynSys} is uniformly observable on $\R^n$ with respect to $\U$ and uniformly reachable from $\xo$. Then it is also incrementally asymptotically stable with respect to $\U$. If in addition system~\eqref{eq:DynSys} is quadratically observable on $\R^n$ with respect to $\U$ and quadratically reachable from $\xo$, then it is incrementally exponentially stable with respect to $\U$.
\end{theorem}

To prove the theorem, we shall need the next proposition, whose proof is given in the appendix.

\begin{proposition}
\label{prop:KLfunBound}
	Let $\rho,\psi \in \Koo$, with $\rho(r) < r$ for any $r > 0$, and let $T > 0$. Then there exists $\beta \in \KL$ such that $\rho^{\lfloor t/T \rfloor}(\psi(r)) \leq \beta(r,t)$, where $\lfloor t/T \rfloor$ denotes the largest integer less than or equal to $t/T$.
\end{proposition}

\begin{proof}[Proof of Theorem~\ref{th:dL2+Obs/Reach->dAS}]
	\noindent\emph{Incremental asymptotic stability} -
	Since system~\eqref{eq:DynSys} is incrementally dissipative, there exists a storage function $S$ such that~\eqref{eq:DynSysAug} is dissipative with respect to the supply rate~\eqref{eq:IncSupRate}. Then, due to uniform observability on $\Rn$ with respect $\U$, there exist $T_o \geq 0$ and $\alpha_o \in \Koo$ such that~\eqref{eq:S-SNegDef} is satisfied for every $T \geq T_o$. Recursive application of~\eqref{eq:S-SNegDef2} from Lemma~\ref{lem:S-SNegDef2} yields
	\begin{equation}
		S\big(x(\tau + kT),\tx(\tau + kT)\big) \leq \rho^k\big(S\big(\xtau,\txtau\big)\big)
	\end{equation}
	for all $\tau \in [0,T)$, $k \in \N$ and $T \geq T_o$, where  $x(\tau + kT) = \phi(\tau + kT, 0, x_0, u_0)$ and $\tx(\tau + kT) = \phi(\tau + kT, 0, \tx_0, u_0)$, with $u_0 \in \U$. Using uniform reachability through Lemma~\ref{lem:SBounds} yields
	\begin{equation}
		S\big(x(\tau + kT),\tx(\tau + kT)\big) \leq \rho^k\big(\eta^2\alpha_r(\norm{\xtau - \txtau})\big)
	\end{equation}
	Lemma~\ref{lem:SNonIncreas} states that $S$ is nonincreasing, which allows us to write
	\begin{equation}
	\label{eq:IncAsStabAux2}
 		\alpha_o(\norm{\xtau - \txtau}) \leq S(\xtau,\txtau) \leq S(x_0,\tx_0) \leq \eta^2\alpha_r(\norm{x_0 - \tx_0})
	\end{equation}
	and then
	\begin{equation}
	\label{eq:IncAsStabAux}
		S\big(x(\tau + kT),\tx(\tau + kT)\big) \leq \rho^k\big(\psi(\norm{x_0 - \tx_0})\big)
	\end{equation}
	with $\psi := \eta^2\alpha_r\circ\alpha_o^{-1}\circ\eta^2\alpha_r \in \Koo$. Let us write $t = \tau + kT$, and note that $k = \lfloor t/T \rfloor$. Then, according to Proposition~\ref{prop:KLfunBound}, there exists $\hat{\beta} \in \KL$ such that
	\begin{equation}
		S\big(\xt,\txt\big) \leq \hat{\beta}(\norm{x_0 - \tx_0},t)
	\end{equation}
	Lastly, using~\eqref{eq:SBounds} from Lemma~\ref{lem:SBounds}, we can write
	\begin{equation}
		\norm{\xt - \txt} \leq \alpha_o^{-1}(\hat{\beta}(\norm{x_0 - \tx_0},t)) =: \beta(\norm{x_0 - \tx_0},t)
	\end{equation}
	with $\beta \in \KL$, which concludes the first part of the proof.

	\noindent\emph{Incremental exponential stability} -
	If system~\eqref{eq:DynSys} is quadratically reachable from $\xo$ and quadratically observable on $\Rn$ with respect to $\U$, we have $\alpha_o(s) = \sigma_os^2$ and $\alpha_r(s) = \sigma_rs^2$. From Lemma~\ref{lem:S-SNegDef2}, we may pick $\rho(s) = (\id - 0.5\alpha_o\circ\eta^2\alpha_r^{-1})(s) = \left(1 - \frac{\sigma_o}{2\eta^2\sigma_r}\right)s =: \mu s$, with $\mu \in (0,1)$. Similarly, $\psi$ becomes $\psi(s) = \sigma_o \hat{d}^2s^2$, with $\hat{d} := \frac{\eta^2\sigma_r}{\sigma_o} > 0$. With this is mind,~\eqref{eq:IncAsStabAux} becomes
	\begin{equation}
		S\big(x(\tau + kT),\tx(\tau + kT)\big) \leq \sigma_o \hat{d}^2 \mu^k\norm{x_0 - \tx_0}^2
	\end{equation}
	We use~\eqref{eq:SQBounds} from Lemma~\ref{lem:SBounds} to write
	\begin{equation}
		\norm{x(\tau + kT) - \tx(\tau + kT)}^2 \leq \hat{d}^2 \mu^k\norm{x_0 - \tx_0}^2
	\end{equation}
	and then
	\begin{equation}
		\norm{x(\tau+kT) - \tx(\tau+kT)} \leq \hat{d}\mu^\frac{k}{2}\norm{x_0 - \tx_0}
	\end{equation}
	Let us write $t := \tau + kT$. For $k \geq 1$, we are able to pick $0 < \lambda \leq -\frac{1}{4T}\log{\mu} \leq -\frac{k}{2(\tau + kT)}\log{\mu}$, to write
	\begin{equation}
	\label{eq:ExpStab}
	 	\norm{\xt - \txt} \leq \hat{d}e^{-\lambda t}\norm{x_0 - \tx_0}
	\end{equation}
	for every $t \geq T$. In view of~\eqref{eq:IncAsStabAux2}, $\norm{\xt - \txt}$ is bounded for $t \in [0,T]$, so that there exists $d \geq \hat{d}$ such that~\eqref{eq:ExpStab} with $\hat{d}$ replaced by $d$ is valid for every $t \geq 0$, which concludes the proof.
\end{proof}

Let us now consider the converse problem, i.e. when does incremental asymptotic stability imply \dLdg stability. This result may be established through the use of the connection between the \dLdg and the \Ldg of the system linearizations around every input $u \in \Ld$. For details, refer to~\cite{Fromion1996a}.

\begin{theorem}
	Assume that $f$ and $h$ in~\eqref{eq:DynSys} are Lipschitz continuous and differentiable functions. Additionally, assume that the jacobian $\partial f/\partial x$ is locally Lipschitz continuous. Under these conditions, if system~\eqref{eq:DynSys} is incrementally asymptotically stable with respect to $\LdepRp$, then it is incrementally \Ldg stable.
\end{theorem}

\begin{proof}
	Incremental asymptotic stability with respect to $\LdepRp$ implies asymptotic stability of every unperturbed motion, as defined in~\cite[Definition 2.5]{Fromion1996a}. Then, the proof is achieved by applying Theorem 2 of~\cite{Fromion1996a}.
\end{proof}

\section{Concluding remarks}
\label{se:ConcRem}

In this \paper we have studied incremental stability properties of PWA systems. We have proposed sufficient conditions allowing to compute an upper bound on the \dLdg and to assess incremental asymptotic stability. The proposed conditions are shown to be less conservative than previous results in the literature, and their application is illustrated through numerical examples. Additionally, the connection between the two properties was highlighted through the use of appropriate concepts of observability and reachability.

The results presented in this \paper open up some perspectives for future research. Firstly, the tools developed in Section~\ref{se:AnalPWA} might be coupled with piecewise-affine approximation techniques (e.g. adapting the results in~\cite{Azuma2010,Zavieh2013}) to extend the analysis to more general nonlinear systems. A second perspective is to study the applicability of these results in the presence of Genuine Zeno behaviors. Finally, we aim to further study the observability of piecewise-affine systems, possibly deepening the connections between the conditions obtained in Section~\ref{se:AnalPWA}.

\appendix

\section{Appendix}
\label{se:Appendix}

\subsection{Lipschitz continuity of PWA systems}
\label{ss:LipCont}

As discussed in Section~\ref{se:IncStabNL}, a sufficient condition for the non-existence of sliding modes is Lipschitz continuity of the right-hand side of~\eqref{eq:PWAsysUY}. The following lemma, adapted from~\cite{Pavlov2007}, ensures continuity, which in turn implies Lipschitz continuity in view of Proposition~\ref{prop:LipContPWA}.

\begin{lemma}
\label{lem:PWALipsCont}
	The right-hand side of~\eqref{eq:PWAsysUY} is continuous if and only if $B_i = B_j$, $\forall i, j \in \I$ and for any two cells $X_i$ and $X_j$ having a common boundary $X_i \cap X_j \subseteq \left\{ x \in X \mid E_\ij x + e_\ij = 0\right\}$ the corresponding matrices $A_i$ and $A_j$ and the vectors $a_i$ and $a_j$ satisfy
	\begin{equation}
	 	\begin{aligned}
	 		gE_\ij &= A_i - A_j \\
	 		ge_\ij &= a_i - a_j 
	 	\end{aligned}
	\end{equation}
	for some vector $g \in \R^n$.
\end{lemma}

\begin{proposition}
\label{prop:LipContPWA}
	If the piecewise-affine function $f(x,u) = A_ix + a_i + Bu$, for $x \in X_i$, is continuous with respect to $x$, then it is also globally Lipschitz continuous with respect to $x$ and $u$.
\end{proposition}
\begin{proof}
	We need to show that there exist $L_x$ and $L_u$ such that
	\begin{equation}
	 	\norm{f(x,u) - f(\tx,\tu)} \leq L_x\norm{x - \tx} + L_u\norm{u - \tu}
	\end{equation}
	The case $x,\tx \in X_i$ is trivial. Let us consider the case where $x\in X_i$ and $\tx\in X_j$, for $i \neq j$. There exists a segment joining $x$ and $\tx$ passing through $r$ regions, and then there exist $r+1$ points $x_0, \ldots, x_r$, with $x_0 = x$, $x_r = \tx$, so that each $x_\ell$, for $\ell \in \{1,\ldots,r-1\}$, lies in the intersection between two regions. Let $\sigma: \{1,\ldots,r\}\rightarrow \I$ be such that $x_\ell \in X_{\sigma(\ell)} \cap X_{\sigma(\ell+1)}$, for $\ell \in \{1,\ldots,r-1\}$. Using continuity and standard norm properties, we may write
	
	\begin{align}
		\norm{f(x,u) - f(\tx,\tu)} &= \norm{A_ix + a_i + Bu - A_j\tx - a_j - B\tu} \nonumber \\
		&\leq \left\|A_{\sigma(1)}x_0 + a_{\sigma(1)} - (A_{\sigma(1)}x_1 + a_{\sigma(1)}) + (A_{\sigma(2)}x_1 + a_{\sigma(2)}) - \cdots\right. \nonumber \\
		&\phantom{\leq}\left. \cdots - (A_{\sigma(r)}x_r + a_{\sigma(r)})\right\| + \norm{B(u - \tu)} \nonumber \\
		&= \norm{A_{\sigma(1)}(x_0 - x_1) + \cdots + A_{\sigma(r)}(x_{r-1} - x_r)} + \norm{B(u - \tu)} \nonumber \\
		&= \norm{\sum_{\ell = 1}^{r}A_{\sigma(\ell)}(x_{\ell - 1} - x_\ell)} + \norm{B(u - \tu)} \nonumber \\
		&\leq \max_\ell \left\{\norm{A_{\sigma(\ell)}}\right\}\sum_{\ell = 1}^{r}\norm{(x_{\ell - 1} - x_\ell)} + \norm{B}\norm{u - \tu} \nonumber \\
		&=: L_x\norm{x - \tx} + L_u\norm{u - \tu}
	\end{align}
	where the last equality comes from the fact that $\lbrace x_{\ell-1} - x_\ell\rbrace_{ \ell \in \{1,\ldots,r\} }$ forms a partition of the line segment joining $x$ and $\tx$, the state partition is finite and matrices $A_i$, $\forall i \in \I$, and $B$ are bounded. This proves the claim.
\end{proof}

\subsection{Proof of proposition~\ref{prop:KLfunBound}}
\label{ss:ProofProp}

We provide here a proof of Proposition~\ref{prop:KLfunBound}. We begin by recalling the next lemma from~\cite{Kellett2014}.

\begin{lemma}
\label{lem:KLfun}
	Suppose a function $\zeta: \Rp \times \Rp \rightarrow \Rp$ satisfies
	\begin{enumerate}[(i)]
		\item for all $s, \epsilon \in \Rpa$, there exists some $\tau = \tau(s,\varepsilon) \in \Rpa$ such that $\zeta(r,t) < \varepsilon$ for all $r \leq s$ and $t \geq \tau$;
		\item for all $\varepsilon \in \Rpa$ there exists $\delta \in \Rpa$ such that $\zeta(r,t) \leq \varepsilon$ for all $r \leq \delta$ and all $t \in \Rp$
	\end{enumerate}
	Then there exists some $\beta \in \KL$ such that $\zeta(r,t) \leq \beta (r,t)$ for all $r,t \in \Rp$
\end{lemma}

We can now state the aforementioned proof.

\begin{proof}[Proof of Proposition~\ref{prop:KLfunBound}]
	It suffices to show that $\zeta(r,t) := \rho^{\lfloor t/T \rfloor}(\psi(r))$ satisfies the conditions in Lemma~\ref{lem:KLfun}.
	
	\noindent\emph{(i)} - We know that $\rho^{k+1}(r) < \rho^k(r)$, $\forall k \in \N$. Then, for a fixed $s$ and $\varepsilon > 0$, there exists $k^\ast$ such that $\rho^{k^\ast}(\psi(s)) \leq \varepsilon$. Let us choose $\tau = k^\ast T$ to write $\zeta(s,\tau) \leq \varepsilon$. Since $\rho,\psi \in \Koo$, the inequality is valid for all $r \leq s$, and since $\rho^{\lfloor t/T \rfloor}$ is non-increasing on $t$, it is valid for every $t \geq \tau$.
	
	\noindent\emph{(ii)} - Since $\rho(r) < r$, the function $\zeta$ is such that $\zeta(r,t) \leq \psi(r)$, for all $t \in \Rp$. Then, since $\psi \in \Koo$, for all $\varepsilon > 0$ we can pick $\delta = \psi^{-1}(\varepsilon)$ such that $\zeta(r,t) \leq \varepsilon$ for all $r \leq \delta$ and all $t \in \Rp$.
	
	This proves the proposition.	
\end{proof}

\bibliographystyle{myhplain}
\bibliography{Ref_arXiv2016}

\begin{thebibliography}{10}

\bibitem{Ames2005}
A.~D. Ames and S.~Sastry.
\newblock Characterization of {Z}eno behavior in hybrid systems using
  homological methods.
\newblock In {\em Proceedings of the 2005 American Control Conference}, pages
  1160--1165, Portland, USA, June 2005.

\bibitem{Angeli2002}
D.~Angeli.
\newblock A {L}yapunov approach to incremental stability properties.
\newblock {\em {IEEE} Transactions on Automatic Control}, 47(3):410--421, 2002.

\bibitem{Angeli2009}
D.~Angeli.
\newblock Further results on incremental input-to-state stability.
\newblock {\em {IEEE} Transactions on Automatic Control}, 54(6):1386--1391,
  June 2009.

\bibitem{Azuma2010}
S.~Azuma, J.~Imura, and T.~Sugie.
\newblock {L}ebesgue piecewise affine approximation of nonlinear systems.
\newblock {\em Nonlinear Analysis: Hybrid Systems}, 4(1):92 -- 102, 2010.

\bibitem{Bemporad2000}
A.~Bemporad, G.~Ferrari-Trecate, and M.~Morari.
\newblock Observability and controllability of piecewise affine and hybrid
  systems.
\newblock {\em {IEEE} Transactions on Automatic Control}, 45(10):1864--1876,
  Oct 2000.

\bibitem{Boyd1994}
S.~Boyd, L.~El~{G}haoui, E.~Feron, and V.~Balakrishnan.
\newblock {\em Linear Matrix Inequalities in System and Control Theory},
  volume~15 of {\em {SIAM} studies in applied mathematics}.
\newblock Society for Industrial and Applied Mathematics, 1994.

\bibitem{Chitour1995}
Y.~Chitour, W.~Liu, and E.~Sontag.
\newblock On the continuity and incremental-gain properties of certain
  saturated linear feedback loops.
\newblock {\em International Journal of Robust and Nonlinear Control},
  5(5):413--440, 1995.

\bibitem{Doban2016}
A.~I. Doban and M.~Lazar.
\newblock Computation of {L}yapunov functions for nonlinear differential
  equations via a {M}assera-type construction.
\newblock {\em ArXiv e-prints}, mar 2016, math.DS/1603.03287.

\bibitem{Fromion1997}
V.~Fromion.
\newblock Some results on the behavior of {L}ipschitz continuous systems.
\newblock In {\em Proceedings of the 1997 European Control Conference}, pages
  2011--2016, Brussels, Belgium, July 1997.

\bibitem{Fromion1995}
V.~Fromion, S.~Monaco, and D.~Normand-Cyrot.
\newblock A possible extension of ${H}_\infty$ control to the nonlinear
  context.
\newblock In {\em Proceedings of the 34th IEEE Conference on Decision and
  Control}, pages 975--980, New Orleans, USA, Dec 1995.

\bibitem{Fromion1996a}
V.~Fromion, S.~Monaco, and D.~Normand-Cyrot.
\newblock A link between input-output stability and {L}yapunov stability.
\newblock {\em Systems \& Control Letters}, 27(4):243--248, 1996.

\bibitem{Fromion2001a}
V.~Fromion, S.~Monaco, and D.~Normand-Cyrot.
\newblock The weighted incremental norm approach: from linear to nonlinear
  ${H}_\infty$ control.
\newblock {\em Automatica}, 37(10):1585--1592, 2001.

\bibitem{Fromion2003}
V.~Fromion and G.~Scorletti.
\newblock A theoretical framework for gain scheduling.
\newblock {\em International Journal of Robust and Nonlinear Control},
  13(10):951--982, 2003.

\bibitem{FromionBook}
V.~Fromion and G.~Scorletti.
\newblock {\em From Engineering Practice to Mathematical Methodology through
  {L}ipschitz Continuity}.
\newblock to appear.

\bibitem{Fromion1999}
V.~Fromion, G.~Scorletti, and G.~Ferreres.
\newblock Nonlinear performance of a {PI} controlled missile: an explanation.
\newblock {\em International Journal of Robust and Nonlinear Control},
  9(8):485--518, 1999.

\bibitem{Georgiou1993}
T.~T. Georgiou.
\newblock Differential stability and robust control of nonlinear systems.
\newblock {\em Mathematics of Control, Signals and Systems}, 6(4):289--306,
  1993.

\bibitem{Hassibi1998}
A.~Hassibi and S.~Boyd.
\newblock Quadratic stabilization and control of piecewise-linear systems.
\newblock In {\em Proceedings of the 1998 American Control Conference}, pages
  3659--3664, Philadelphia, USA, 1998.

\bibitem{Hill1980a}
D.~J. Hill and P.~J. Moylan.
\newblock Connections between finite-gain and asymptotic stability.
\newblock {\em {IEEE} Transactions on Automatic Control}, 25(5):931--936, Oct
  1980.

\bibitem{James1993}
M.~R. James.
\newblock A partial differential inequality for dissipative nonlinear systems.
\newblock {\em Systems \& Control Letters}, 21(4):315--320, 1993.

\bibitem{James1995}
M.~James and S.~Yuliar.
\newblock Numerical approximation of the ${H}_\infty$ norm for nonlinear
  systems.
\newblock {\em Automatica}, 31(8):1075--1086, 1995.

\bibitem{Johansson2003}
M.~Johansson.
\newblock {\em Piecewise Linear Control Systems: A Computational Approach},
  volume 284 of {\em Lecture Notes in Control and Information Sciences}.
\newblock Springer Berlin Heidelberg, 2003.

\bibitem{Johansson1998}
M.~Johansson and A.~Rantzer.
\newblock Computation of piecewise quadratic {L}yapunov functions for hybrid
  systems.
\newblock {\em {IEEE} Transactions on Automatic Control}, 43(4):555--559, 1998.

\bibitem{Kellett2014}
C.~M. Kellett.
\newblock A compendium of comparison function results.
\newblock {\em Mathematics of Control, Signals, and Systems}, 26(3):339--374,
  2014.

\bibitem{Khan2015}
K.~A. Khan and P.~I. Barton.
\newblock Switching behavior of solutions of ordinary differential equations
  with abs-factorable right-hand sides.
\newblock {\em Systems \& Control Letters}, 84:27--34, 2015.

\bibitem{Lohmiller1998}
W.~Lohmiller and J.-J.~E. Slotine.
\newblock On contraction analysis for non-linear systems.
\newblock {\em Automatica}, 34(6):683--696, 1998.

\bibitem{Pavlov2007}
A.~Pavlov, A.~Pogromsky, N.~van~de Wouw, and H.~Nijmeijer.
\newblock On convergence properties of piecewise affine systems.
\newblock {\em International Journal of Control}, 80(8):1233--1247, 2007,
  http://dx.doi.org/10.1080/00207170701261978.

\bibitem{Pavlov2006}
A.~V. Pavlov, N.~van~de Wouw, and H.~Nijmeijer.
\newblock {\em Uniform Output Regulation of Nonlinear Systems - A Convergent
  Dynamics Approach}.
\newblock Systems \& Control: Foundations \& Applications. Birkh\"auser Basel,
  1 edition, 2006.

\bibitem{Rantzer2000}
A.~Rantzer and M.~Johansson.
\newblock Piecewise linear quadratic optimal control.
\newblock {\em {IEEE} Transactions on Automatic Control}, 45(4):629--637, Apr
  2000.

\bibitem{Romanchuk1996}
B.~G. Romanchuk and M.~R. James.
\newblock Characterization of the ${L}_p$ incremental gain for nonlinear
  systems.
\newblock In {\em Proceedings of the 35th IEEE Conference on Decision and
  Control}, pages 3270--3275, Kobe, Japan, 1996.

\bibitem{Romanchuk1999}
B.~G. Romanchuk and M.~C. Smith.
\newblock Incremental gain analysis of piecewise linear systems and application
  to the antiwindup problem.
\newblock {\em Automatica}, 35(7):1275 -- 1283, 1999.

\bibitem{Thuan2014}
L.~Q. Thuan.
\newblock Non-{Z}enoness of piecewise affine dynamical systems and affine
  complementarity systems with inputs.
\newblock {\em Control Theory and Technology}, 12(1):35--47, 2014.

\bibitem{Vidyasagar1993}
M.~Vidyasagar.
\newblock {\em Nonlinear Systems Analysis}.
\newblock Englewood Cliffs: Prentice Hall, second edition, 1993.

\bibitem{Waitman2016}
S.~Waitman, P.~Massioni, L.~Bako, G.~Scorletti, and V.~Fromion.
\newblock Incremental $\mathcal{L}_2$-gain analysis of piecewise-affine systems
  using piecewise quadratic storage functions.
\newblock In {\em Proceedings of the 55th IEEE Conference on Decision and
  Control}, Las Vegas, USA, 2016.

\bibitem{Willems1971b}
J.~C. Willems.
\newblock The generation of {L}yapunov functions for input-output stable
  systems.
\newblock {\em {SIAM} Journal on Control}, 9(1):105--134, 1971,
  http://dx.doi.org/10.1137/0309009.

\bibitem{Willems1972a}
J.~C. Willems.
\newblock Dissipative dynamical systems part {I}: General theory.
\newblock {\em Archive for Rational Mechanics and Analysis}, 45(5):321--351,
  1972.

\bibitem{Zames1963}
G.~Zames.
\newblock Functional analysis applied to nonlinear feedback systems.
\newblock {\em {IEEE} Transactions on Circuit Theory}, 10(3):392--404,
  September 1963.

\bibitem{Zames1966b}
G.~Zames.
\newblock On the input-output stability of time-varying nonlinear feedback
  systems---part {II}: Conditions involving circles in the frequency plane and
  sector nonlinearities.
\newblock {\em {IEEE} Transactions on Automatic Control}, 11(3):465--476, Jul
  1966.

\bibitem{Zames1966a}
G.~Zames.
\newblock On the input-output stability of time-varying nonlinear feedback
  systems---part {I}: Conditions derived using concepts of loop gain, conicity,
  and positivity.
\newblock {\em {IEEE} Transactions on Automatic Control}, 11(2):228--238, 1966.

\bibitem{Zavieh2013}
A.~Zavieh and L.~Rodrigues.
\newblock Intersection-based piecewise affine approximation of nonlinear
  systems.
\newblock In {\em Proceedings of the 21st Mediterranean Conference on Control
  Automation}, pages 640--645, Platanias-Chania, Greece, 2013.

\end{thebibliography}

\end{document}